%% file: neurips_2018.tex
\documentclass{article}
\usepackage{natbib}

\usepackage[final]{neurips_2018}

\usepackage[utf8]{inputenc} 
\usepackage[T1]{fontenc}    
\usepackage{hyperref}       
\usepackage{url}            
\usepackage{booktabs}       
\usepackage{amsfonts}       
\usepackage{nicefrac}       
\usepackage{microtype}      

\newcommand{\lh}[1]{}\newcommand{\ab}[1]{}\newcommand{\mj}[1]{}

\usepackage[textwidth=3.3cm]{todonotes}
\renewcommand{\lh}[1]{\todo[color=red!20,size=\footnotesize]{LH: #1}{}}
\renewcommand{\ab}[1]{\todo[color=purple!20,size=\footnotesize]{AB: #1}{}}
\renewcommand{\mj}[1]{\todo[color=green!20,size=\footnotesize]{MJ: #1}{}}

\newcommand*\samethanks[1][\value{footnote}]{\footnotemark[#1]}

\input{defs_nips18}

\usepackage{thm-restate}
\title{
\cola:
Decentralized Linear Learning
}

\author{
Lie He
\thanks{These two authors contributed equally}
\\
EPFL\\
\texttt{lie.he@epfl.ch}
\And
{An Bian\samethanks[1] } \thanks{Now known as Yatao A. Bian. ORCID: \href{https://orcid.org/0000-0002-2368-4084}{orcid.org/0000-0002-2368-4084}}
\\
ETH Zurich\\
\texttt{ybian@inf.ethz.ch}
\And
Martin Jaggi\\
EPFL\\
\texttt{martin.jaggi@epfl.ch}
}

\begin{document}

\maketitle

\begin{abstract}
Decentralized machine learning is a promising emerging paradigm in view of global challenges of data ownership and privacy. We consider learning of linear classification and regression models, in the setting where the training data is decentralized over many user devices, and the learning algorithm must run on-device, on an arbitrary communication network, without a central coordinator.
We propose \cola, a new decentralized training algorithm with strong theoretical guarantees and superior practical performance. Our framework overcomes many limitations of existing methods, and achieves communication efficiency, scalability, elasticity as well as resilience to changes in data and allows for unreliable and heterogeneous participating devices.
\end{abstract}

%

\section{Introduction}

With the immense growth of data, decentralized machine learning has become not only attractive but a necessity. 
Personal data from, for example, smart phones, wearables and many other mobile devices is sensitive and exposed to a great risk of data breaches and abuse when collected by a centralized authority or enterprise. 
Nevertheless, many users have gotten accustomed to giving up control over their data in return for useful machine learning predictions (e.g. recommendations), which benefits from joint training on the data of all users combined in a centralized fashion.

In contrast, decentralized learning aims at learning this same global machine learning model, without any central server. Instead, we only rely on distributed computations of the devices themselves, with each user's data never leaving its device of origin.
While increasing research progress has been made towards this goal, major challenges in terms of the privacy aspects as well as algorithmic efficiency, robustness and scalability remain to be addressed.
Motivated by aforementioned challenges, we make progress in this work addressing the important problem of training generalized linear models in a fully decentralized environment.

Existing research on decentralized optimization, $\min_{\xv\in\R^n} F(\xv)$, can be categorized into two main directions. The seminal line of work started by Bertsekas and Tsitsiklis in the 1980s, cf. \citep{Tsitsiklis:1986ee}, tackles this problem by splitting the parameter vector $\xv$ by coordinates/components among the devices.
A second more recent line of work including e.g. \citep{nedic2009distributed,duchi2012ddual,shi2015extra,mokhtari2016dsa,nedic2017achieving}
 addresses sum-structured  $F(\xv) = \sum_k F_k(\xv)$  where $F_k$ is the local cost function of node $k$. This structure is closely related to empirical risk minimization in a learning setting. See e.g. \citep{cevher2014review} for an overview of both directions. 
While the first line of work typically only provides convergence guarantees for smooth objectives~$F$, the second approach often suffers from a ``lack of consensus'', that is, the minimizers of $\{F_k\}_k$ are typically different since the data is not distributed i.i.d. between devices in general.

\paragraph{Contributions.}
In this paper, our main contribution is to propose \cola, a new decentralized framework for training generalized linear models with convergence guarantees.
Our scheme resolves both described issues in existing approaches, using techniques from primal-dual optimization, and can be seen as a generalization of \cocoa \citep{smith2016cocoa} to the decentralized setting.
More specifically, the proposed algorithm offers \vspace{-1mm}
\begin{itemize}
\item \emph{Convergence Guarantees:} Linear and sublinear convergence rates are guaranteed for strongly convex and general convex objectives respectively.  
Our results are free of the restrictive assumptions made by stochastic methods  \citep{zhang2015elastic,Wang:2017th}, which requires i.i.d. data distribution over all devices. 
\item \emph{Communication Efficiency and Usability:} Employing a data-local subproblem between each communication round, \cola not only achieves communication efficiency but also allows the re-use of existing efficient single-machine solvers for on-device learning. 
We provide practical decentralized primal-dual certificates to diagnose the learning progress.
\item \emph{Elasticity and Fault Tolerance:} 
Unlike sum-structured approaches such as SGD
, \cola is provably resilient to changes in the data, in the network topology, and participating devices disappearing, straggling or re-appearing in the network.
\end{itemize}

Our implementation is publicly available under \href{https://github.com/epfml/cola}{github.com/epfml/cola} .

%
%

\subsection{Problem statement}

\textbf{Setup.}
Many machine learning and signal processing models are formulated as a composite convex optimization problem of the form
\[
\min_{\u}  \;  l(\u) + r(\u),
\]
where $l$ is a convex loss function of a linear predictor over data and $r$ is a convex regularizer. Some cornerstone applications include e.g. logistic regression, SVMs, Lasso, generalized linear models, each combined with or without L1, L2 or elastic-net regularization.
Following the setup of \citep{Dunner:2016vga,smith2016cocoa}, these training problems can be mapped to either of the two following formulations, which are dual to each other
\begin{align}
	& \min_{\xv\in\R^n} \big[\OA(\xv) \eqdef f(\BA\xv) + \textstyle\sum_i g_i(x_i)\ \big] \label{eq_A} \tag{\color{link_color}{A}}\\
	& \min_{\wv\in\R^d} \big[\OB(\w) \eqdef \cconj{f}(\w) + \textstyle\sum_i \cconj{g_i}(- \BA_i^\trans \w)\big], \label{eq_B} \tag{\color{link_color}{B}}
\end{align}
where $\cconj{f}, \cconj{g}_i$ are the convex conjugates of $f$ and $g_i  $, respectively. Here
$\xv\in\R^n$ is a parameter vector and $\BA\eqdef [\BA_1; \ldots; \BA_n] \in\R^{d\times n}$ is a data matrix with column vectors $\BA_i\in\R^d, i\in[n]$. We assume that $f$ is smooth (Lipschitz gradient) and $g(\xv) := \sum_{i=1}^ng_i(x_i)$ is \emph{separable}.
%
%

\textbf{Data partitioning.}  As in \citep{jaggi2014communication,Dunner:2016vga,smith2016cocoa}, we assume the dataset~$\BA$ is distributed over~$K$ machines according to a partition $\{\Pk\}_{k=1}^K$ of the \textit{columns} of $\BA$. Note that 
this convention maintains the flexibility of partitioning the training dataset 
either by samples (through mapping applications to~\eqref{eq_B}, e.g. for SVMs) or by features (through mapping applications to \eqref{eq_A}, e.g. for Lasso or L1-regularized logistic regression). 
For $\xv\in \R^n$, we write $\vsubset{\xv}{k}\in\R^n$ for the $n$-vector with elements $(\vsubset{\xv}{k})_i:=x_i$ if $i\in\Pk$ and $(\vsubset{\xv}{k})_i:=0$ otherwise, and analogously $\vsubset{\BA}{k}\in\R^{d\times n_k}$ for the corresponding set of local data columns on node $k$, which is of size $n_k=|\Pk|$.

\textbf{Network topology.}
We consider the task of joint training of a global machine learning model in a {decentralized} network of  $K$ nodes. Its connectivity is modelled by a mixing matrix $\mixingmat\in \R_+^{K\times K}$.
More precisely, $\mixingmat_{ij} \in [0,1]$ denotes the connection strength
between nodes $i$ and $j$, with a non-zero weight indicating the existence of a pairwise communication link.
We assume $\mixingmat$ to be symmetric and
doubly stochastic, which means each row and column of~$\mixingmat$ sums to
one.

The spectral properties of $\mixingmat$ used in this paper are that the
eigenvalues of $\mixingmat$ are real, and $1 = \lambda_1(\mixingmat) \geq \cdots \geq  \lambda_n(\mixingmat)  \geq -1$.  Let the second largest
magnitude of the eigenvalues of $\mixingmat$ be
$\beta := \max \{ |\lambda_2(\mixingmat)|, | \lambda_n(\mixingmat) | \}$.
$1-\beta$ is called the \textit{spectral gap}, a quantity well-studied in graph theory and network analysis. 
The spectral gap measures the level of connectivity among nodes. In the extreme case when $\mixingmat$ is diagonal, and thus an identity matrix, the spectral gap is 0 and there is no communication among nodes.
To ensure convergence of decentralized algorithms, we impose the standard assumption of positive spectral gap of the network which includes all connected graphs, such as e.g. a ring or 2-D grid topology, see also \cref{sec:graphTopology} for details.



\subsection{Related work}

Research in decentralized optimization dates back to the 1980s with the seminal work of Bertsekas and Tsitsiklis, cf. \citep{Tsitsiklis:1986ee}.
Their framework focuses on the minimization of a (smooth) function by distributing the components of the parameter vector $\xv$ among agents. In contrast, a second more recent line of work \citep{nedic2009distributed,duchi2012ddual,shi2015extra,mokhtari2016dsa,nedic2017achieving,icmlScamanBBLM17,scaman2018optimal} considers minimization of a sum of individual local cost-functions $F(\xv) = \sum_i F_i(\xv)$, which are potentially non-smooth. Our work here can be seen as bridging the two scenarios to the primal-dual setting \eqref{eq_A} and~\eqref{eq_B}.


While decentralized optimization is a relatively mature area in the operations research and automatic control communities, it has recently received a surge of attention for machine learning applications, see e.g. \citep{cevher2014review}.
Decentralized gradient descent (DGD) with diminishing stepsizes was proposed by~\citep{nedic2009distributed,jakovetic2012convergence}, showing convergence to the optimal solution at a sublinear rate. 
 \citep{yuan2016convergence} further prove that DGD will 
converge to the neighborhood of a global optimum at a linear rate when used with a  constant stepsize for strongly convex objectives.
\citep{shi2015extra} present EXTRA, which offers a significant performance
boost compared to DGD by using a gradient tracking technique.
\citep{nedic2017achieving} propose the \diging algorithm
to handle a time-varying network topology. 
For a static and symmetric $\mixingmat$, \diging recovers EXTRA by redefining the two mixing matrices in EXTRA.
The dual averaging method~\citep{duchi2012ddual} converges at a sublinear rate with a dynamic stepsize.
Under a strong convexity assumption, decomposition techniques such as
decentralized ADMM (DADMM, also known as consensus ADMM) have linear convergence for time-invariant undirected graphs, if subproblems are solved exactly~\citep{Shi:2014js,Wei:2013wy}.
DADMM+ \citep{bianchi2016coordinate} is a different primal-dual approach with more efficient closed-form updates in each step (as compared to ADMM), and is proven to converge but without a rate. Compared to \cola, neither of DADMM and DADMM+ can be flexibly adapted to the communication-computation tradeoff due to their fixed update definition, and both require additional hyperparameters to tune in each use-case (including the $\rho$ from ADMM).
Notably \cola shows superior performance compared to DIGing  and decentralized ADMM in our experiments.
\citep{icmlScamanBBLM17,scaman2018optimal} present lower complexity bounds and optimal algorithms for 
objectives in the form $F(\xv) = \sum_i F_i(\xv)$. Specifically,   \citep{icmlScamanBBLM17}  assumes each $F_i(\xv)$ is 
smooth and strongly convex, and~\citep{scaman2018optimal} assumes each $F_i(\xv)$ is 
Lipschitz continuous and  convex. 	Additionally~\citep{scaman2018optimal} needs a boundedness constraint for the input problem.  In contrast, \cola can handle non-smooth and non-strongly convex objectives \eqref{eq_A} and \eqref{eq_B}, suited to the mentioned applications in machine learning and signal processing.
%
%
For smooth nonconvex models, 
\citep{lian2017can} demonstrate that a variant of decentralized parallel SGD can outperform the centralized variant when the network latency is high. They further extend it to the asynchronous
setting~\citep{lian2017asynchronous} and to deal with large data variance
among nodes~\citep{tang2018d} or with unreliable network links~\citep{tang2018decentralized}. 
For the decentralized, asynchronous consensus optimization, \citep{wu2018decentralized} extends the existing PG-EXTRA and proves convergence of the algorithm. \citep{sirb2018decentralized} proves a $O(K/\epsilon^2)$ rate for stale and stochastic gradients. \citep{lian2017asynchronous} achieves $O(1/\epsilon)$ rate and has linear speedup with respect to number of workers.

In the distributed setting with a central server, algorithms of the \cocoa family  \citep{yang2013disdca,jaggi2014communication,ma2015adding,duenner2018trust}---see \citep{smith2016cocoa} for a recent overview---are targeted for problems of the forms \eqref{eq_A} and \eqref{eq_B}. For convex models, \cocoa has shown to significantly outperform competing methods including e.g., ADMM, 
distributed SGD etc.
Other centralized algorithm representatives are parallel SGD variants such as~\citep{agarwal2011distributed,zinkevich2010parallelized} and more recent
distributed second-order methods \citep{zhang2015disco,reddi2016aide,gargiani2017master,lee2017distributed,duenner2018trust,lee2018distributed}.

In this paper we extend \cocoa to the challenging \textit{decentralized} environment---with no central coordinator---while maintaining all of its nice properties.
We are not aware of any existing primal-dual methods in the decentralized setting, 
except the recent work of \citep{smith2017federated} on federated learning for the special case of multi-task learning problems.
Federated learning was first described by  \citep{konevcny2015federated,konevcny2016federated,mcmahan2017communication} as decentralized learning for on-device learning  applications, combining a global shared model with local personalized models.
Current federated optimization algorithms (like FedAvg in  \citep{mcmahan2017communication}) are still close to the centralized setting. In contrast, our work provides a fully decentralized
alternative algorithm for federated learning with generalized linear models.

\section{The decentralized algorithm: \cola}\label{sec:cola}

The \cola framework is summarized  in \cref{alg_dcocoa}. For a given input problem we map it to either of the \eqref{eq_A} or \eqref{eq_B} formulation,
and define the locally stored dataset $\vsubset{\BA}{k}$ and local part of the weight vector $\vsubset{\x}{k}$ in node $k$ accordingly. 
While $\vv=\BA\x$ is the shared state being communicated in \cocoa, this is generally unknown to a node in the fully decentralized setting. Instead, we maintain~$\vv_k$, a local estimate of $\vv$ in node $k$, and use it as a surrogate in the algorithm. 

\begin{algorithm}[t]
	\caption{\cola: \emphtitle{Co}mmunication-Efficient Decentralized \emphtitle{L}inear Le\emphtitle{a}rning}
	\label{alg_dcocoa}
	\textbf{Input}: Data matrix $\BA$ distributed column-wise according to partition $\{\Pk\}_{k=1}^K$. Mixing matrix $\mixingmat$.
	Aggregation parameter $\aggpar\!\in\![0,1]$,
		and local subproblem parameter $\sigma'$ 
		as in  \labelcref{eq_new_subproblem}.
		Starting point $\vc{\xv}{0} := \0 \in \R^n$, $\vc{\vv}{0}:=\0\in \R^d$, $\vc{\vv}{0}_k:=\0\in \R^d~ \forall~ k=1,\ldots K$;

	\BlankLine
	\For {$t = 0, 1, 2, \dots, T$}{
		\For {$k \in \{1,2,\dots,K\}$ {\bf in parallel over all nodes}}{

			compute locally averaged shared vector
				$\vc{\vv}{t+\frac{1}{2}}_k   :=
				\textstyle \sum_{l=1}^K \mixingmat_{kl} \vc{\vv}{t}_l $\\

			$\vsubset{\Delta \xv}{k}\leftarrow$ $\Theta$-approximate solution
				to subproblem~\eqref{eq_new_subproblem} at $\vc{\vv}{t+\frac{1}{2}}_k$\\[-1mm]
			update local variable $\vsubset{\vc{\xv}{t+1}}{k} := \vsubset{\vc{\xv}{t}}{k} + \aggpar \, \vsubset{\Delta \xv}{k}$\\
			compute  update of   local    estimate    $\Delta \vv_k :=  \BA_{[k]}
				\vsubset{\Delta \xv}{k}$\\


			$\vc{\vv}{t+1}_k := \vc{\vv}{t+\frac{1}{2}}_k  + \gamma K\Delta \vv_k$ \label{step_strategy2}

		}
	}
\end{algorithm}


\bigskip
\textbf{New data-local quadratic subproblems.}
During a computation step, node $k$ locally solves the following minimization problem\vspace{-2mm}
\begin{equation}
\label{eq_new_subproblem} 
\min_{\vsubset{\Delta \xv}{k}\in\R^{\n}} \ %
\newsub(\vsubset{\Delta\xv}{k};\v_k, \vsubset{\xv}{k}), \vspace{-2mm}
\end{equation}
where\vspace{-2mm}
\begin{equation}\label{eq:subproblem_modified_1}
	\begin{split}\textstyle
		\newsub(\vsubset{\Delta\xv}{k};\v_k, \vsubset{\xv}{k})
		:=\ &
		\textstyle
		\frac{1}{K}
			f(\v_k)
			+ 		\nabla f(\v_k) ^\top
					\vsubset{\BA}{k}\vsubset{\Delta\xv}{k}
				\\
			&
			\textstyle
			+ \frac{\sigma'}{2\tau} \norm{\vsubset{\BA}{k}\vsubset{\Delta\xv}{k}}^2
			+ \sum_{i \in \Pk} g_i(x_i + {(\vsubset{\Delta \xv}{k})}_i).
	\end{split}
\end{equation}
Crucially, this subproblem only depends on the local data $\vsubset{\BA}{k}$, and local vectors~$\vv_l$ from the neighborhood of the current node $k$.
In contrast, in \cocoa \citep{smith2016cocoa} the subproblem is defined in terms of a global aggregated shared vector $\vv_c := \BA\xv \in\R^{\d}$, which is not available in the decentralized setting.\footnote{%
\emph{Subproblem interpretation:} Note that for the special case of $\aggpar:=1$, $\sigma':= K$, by smoothness of $f$, our subproblem in \eqref{eq:subproblem_modified_1} is an upper bound on
\begin{equation}\label{eq:subproblem_2}\textstyle
\min_{\vsubset{\Delta \xv}{k}\in\R^{\n}} 
\frac{1}{K} f(\BA(\x + K \vsubset{\Delta \xv}{k})) + \sum_{i\in\Pk} g_i(x_i + {(\vsubset{\Delta \xv}{k})}_i),
\end{equation}
which is a scaled block-coordinate update of block $k$ of the original objective \eqref{eq_A}. This assumes that we have consensus $\v_k\equiv\BA\x$ $\forall~k$.
For \textit{quadratic} objectives (i.e. when $f\equiv\norm{.}_2^2$ and~$\BA$ describes the quadratic), the equality of the formulations \eqref{eq:subproblem_modified_1} and \eqref{eq:subproblem_2} holds.
Furthermore, by convexity of $f$, the sum of \eqref{eq:subproblem_2} is an upper bound on the centralized updates $f(\x+ \Delta\x) + g(\x + \Delta\x)$. 
Both inequalities quantify the overhead of the distributed algorithm over the centralized version, see also \citep{yang2013disdca,ma2015adding,smith2016cocoa} for the non-decentralized case.
}
The aggregation parameter $\aggpar\in[0, 1]$ does not need to be tuned; in fact, we use the default $\aggpar:=1$ throughout the paper, see \citep{ma2015adding} for a discussion.
Once $\aggpar$ is settled, a safe choice of the subproblem relaxation parameter $\sigma'$ is given as $\sigma':=\aggpar K$. 
$\sigma'$ can be additionally tightened using 
an improved Hessian subproblem (\cref{ssec:hessian}).



\textbf{Algorithm description.} At time $t$ on node $k$, $\vc{\vv}{t+\frac{1}{2}}_k$ is a local estimate of  
the shared variable after a communication step (i.e. gossip mixing). The local subproblem \eqref{eq_new_subproblem} based on this estimate is solved and yields
$\vsubset{\Delta \xv}{k}$.
%
Then we calculate $\Delta \vv_k :=  \vsubset{\BA}{k} \vsubset{\Delta \xv}{k}$, and update the local shared vector 
 $\vc{\vv}{t+1}_k$. 
%
%
We allow the local subproblem to be solved approximately:
\begin{assumption}[$\Theta$-approximation solution]\label{assumption:theta}
	Let $\Theta\in[0, 1]$ be the \emph{relative accuracy} of the local solver (potentially randomized), in the sense of returning an approximate solution $\Delta\vsubset{\xv}{k}$ at each step $t$, s.t.\vspace{-1mm}
	\begin{equation*}
	\frac{
	\mathbb{E}[
	\newsub(\Delta\vsubset{\xv}{k}; \v_k, \vsubset{\xv}{k})
	-\newsub(\Delta\vsubset{\xv^\star}{k}; \v_k, \vsubset{\xv}{k})
	]
	}{
	~~\newsub(~~\0~~~~; \v_k, \vsubset{\xv}{k})
	~-\newsub(\Delta\vsubset{\xv^\star}{k}; \v_k, \vsubset{\xv}{k})
	}
	 \le
	\Theta,
	\end{equation*}
	where
	$\Delta\vsubset{\xv^\star}{k}\in\argmin_{\Delta\xv\in\R^n} \newsub(\Delta\vsubset{\xv}{k}; \v_k, \vsubset{\xv}{k}), \text{ for each } k \in [K]$.
\end{assumption}

\paragraph{Elasticity to network size, compute resources and changing data---and fault tolerance.}\label{par:adaptivity}
Real-world communication networks are not homogeneous and static, but greatly vary in availability, computation, communication and storage capacity. Also, the training data is subject to changes. While these issues impose significant challenges for most existing distributed training algorithms, we hereby show that \cola offers adaptivity to such dynamic and heterogenous scenarios. 

Scalability and elasticity in terms of availability and computational capacity can be modelled by a node-specific local accuracy parameter $\Theta_k$ in \cref{assumption:theta}, as proposed by \citep{smith2017federated}. The more resources node $k$ has, the more accurate (smaller) $\Theta_k$ we can use.
The same mechanism also allows dealing with fault tolerance and stragglers, which is crucial e.g. on a network of personal devices.
More specifically, when a new node $k$ joins the network, its $\vsubset{\xv}{k}$ variables are initialized to~$\0$; when node $k$ leaves, its $\vsubset{\xv}{k}$ is frozen, and its subproblem is not touched anymore (i.e. $\Theta_k=1$). Using the same approach, we can adapt to dynamic changes in the dataset---such as additions and removal of local data columns---by adjusting the size of the local weight vector accordingly. Unlike gradient-based methods and ADMM, \cola does not require parameter tuning to converge, increasing resilience to drastic changes. 

\textbf{Extension to improved second-order subproblems.}
In the centralized setting, it has recently been shown that the Hessian information  of $f$ can be properly utilized to define improved local subproblems \citep{lee2017distributed,duenner2018trust}.
Similar techniques can be applied to \cola as well, details on which are left in \cref{app_extensions}.

\textbf{Extension to time-varying graphs.}
Similar to scalability and elasticity, it is also straightforward to extend \cola to a time varying graph under proper assumptions. If we use the time-varying model in \cite[Assumption 1]{nedic2017achieving}, where an undirected graph is connected with $B$ gossip steps, then changing \cola to perform $B$ communication steps and one computation step per round still guarantees convergence. Details of this setup are provided in \cref{app_extensions}.



\section{On the convergence of \cola}
In this section we present a convergence analysis of the proposed decentralized
algorithm \cola for both general convex and strongly convex objectives.
In order to capture the evolution of \cola, we reformulate the original problem \eqref{eq_A} by incorporating both $\x$ and local estimates $\{\vv_k\}_{k=1}^K$
\begin{align}\label{eq:H}
& \textstyle
\min_{ \xv, \{\vv_k\}_{k=1}^K }
 \ooa(\xv, \{\vv_k\}_{k=1}^K)
:=
\frac1K \sum_{k=1}^{K}
f(\vv_k)
+
g(\xv),   \qquad \tag{\color{link_color}{DA}}
\\\notag 
& \text{ such that } \qquad  \vv_k = \BA \x, \;  k=1,...,K. 
\end{align}
While the consensus is not always satisfied during
 \cref{alg_dcocoa},
the following relations between the decentralized objective and the original one \eqref{eq_A} always hold. All proofs are deferred to Appendix~\ref{sec:proofs}.
%
%
%
%
\begin{restatable}{lemma}{restatlemmaDalphaHalpha}
\label{lemma:DH}
Let $\{\v_k\}$ and $\x$ be the iterates generated during the execution of \cref{alg_dcocoa}. At any timestep, it holds that\vspace{-2mm}
%
\begin{align}\label{eq:DH:1}
& \textstyle	
\frac{1}{K}\sum_{k=1}^{K}{\vv_k} = \BA\xv,\\
& \textstyle	
\OA(\xv) \le \ooa (\xv, \{\vv_k\}_{k=1}^K)\le
\OA(\xv)
+
\frac{1}{2\tau K}\sum  _{k=1}^{K} \norm{\vv_k- \BA\x}^2.
\end{align}
\end{restatable}
The dual problem and duality gap of the decentralized objective \eqref{eq:H} are given in \cref{definition:GH}.
\begin{restatable}[Decentralized Dual Function and Duality Gap]{lemma}{restaDecenDual}\label{definition:GH} 
The Lagrangian dual of the decentralized formation \eqref{eq:H} is	
\begin{equation}\label{decentralized_dual} \textstyle
\min_{ \{\w_k \}_{k=1}^K}  
\oob ({\{\w_k \}_{k=1}^K})
:= \frac{1}{K} \sum_{k=1}^{K}  f^*(\w_k) 
+ \sum_{i =1}^n g^*_i \!\left(-  \BA_i^\top (\frac{1}{K}\sum_{k=1}^{K}\w_k ) \right).
\tag{\color{link_color}{DB}}
\end{equation}
Given primal variables $ \{  \x, \{\v_k \}_{k=1}^K \}$ and dual variables $ \{\w_k\}_{k=1}^K$, the duality gap is:
\begin{equation}\label{eq:GH} \textstyle
	\colagap (\x, \{\v_k \}_{k=1}^K, \{\w_k \}_{k=1}^K )  :=
	\frac{1}{K} \sum_k (f(\v_k) \!+\! f^*(\w_k)) + 
	g(\x)
	\!+\! \sum_{i =1}^n g^*_i \!\left( - \frac{1}{K}\sum_k  \BA_i^\top \w_k  \right).
\end{equation}
\end{restatable}
If the dual variables are fixed to the optimality condition $\w_k = \nabla f(\v_k)$, then the dual variables can be omitted in the argument list of duality gap, namely $\colagap (\x, \{\v_k \}_{k=1}^K)$.
Note that the decentralized duality gap generalizes the duality  gap of \cocoa: when consensus is ensured, i.e., $\v_k\equiv\BA\x$ and $\w_k\equiv\nabla f(\BA\x)$, the decentralized duality gap recovers that of \cocoa.

\subsection{Linear rate for strongly convex objectives}




We use the  following data-dependent quantities in our main theorems
\begin{align}\label{eq:sigma}\textstyle
\sigma_k:= \max_{\vsubset{\xv}{k}\in\R^n}
\norm{\BA_{[k]}\vsubset{\xv}{k}}^2 / \|\vsubset{\xv}{k}\|^2, ~
\sigma_\text{max} = \max_{k=1,\ldots, K}{\sigma_k}, ~
\sigma := \sum_{k=1}^{K} \sigma_k n_k.
\end{align}
If $\{g_i\}$ are  strongly convex, \cola achieves  the following linear rate of  convergence.
\begin{restatable}[Strongly Convex $g_i$]{theorem}{TheoremLinearRate}\label{theorem:linear_rate}
	Consider \cref{alg_dcocoa} with $\gamma:=1$ and let $\Theta$ be the quality of the local solver in \cref{assumption:theta}.
	Let $g_i$ be $\mu_g$-strongly convex for all $i\in[n]$ and let $f$ be $1/\tau$-smooth.
	Let $\bar{\sigma}':=(1 + \beta)\sigma'$, $\alpha:=(1+\frac{(1-\beta)^2}{36(1+\Theta)\beta})^{-1}$ and $\eta:=\gamma(1-\Theta)(1-\alpha)$
	\begin{align} \textstyle
	s_0=
	\frac{\tau\mu_g}{\tau\mu_g+\sigma_\text{max} \bar{\sigma}'}\in [0, 1].
	\end{align}
	Then after $T$ iterations of \cref{alg_dcocoa} with\footnote{$\vc{\varepsilon}{0}_\hh := \ooa(\vc{\x}{0},  \{\vc{\v_k}{0} \}_{k=1}^K) - \ooa(\x^\star,  \{{\v_k^\star} \}_{k=1}^K)$ is the initial suboptimality.}
	\begin{equation*}
	\textstyle
	T \ge  \frac{1 + \eta s_0}{\eta s_0}
	\log \frac{\vc{\varepsilon}{0}_\hh}{\varepsilon_\hh}, 
	\end{equation*}
	it holds that $\E\big[\ooa(\vc{\x}{T},  \{\vc{\v_k}{T} \}_{k=1}^K) - \ooa(\x^\star,  \{{\v_k^\star} \}_{k=1}^K)\big] \le \varepsilon_\hh$.
	Furthermore, after $T$ iterations with\vspace{-2mm}
	\begin{align*}\textstyle
	T \ge  \frac{1 + \eta s_0}{\eta s_0}
	\log \left(
	\frac{1}{\eta s_0}
	\frac{\vc{\varepsilon_\hh}{0}}{\varepsilon_{G_\hh}},
	\right)
	\end{align*}
	we have the expected duality gap $\E[G_\hh(\vc{\x}{T}, \{\vc{\sum_{k=1}^K\mixingmat_{kl}\v_l}{T} \}_{k=1}^K)] \le \varepsilon_{G_\hh}$.
\end{restatable}

\subsection{Sublinear rate for general convex objectives}

Models such as sparse logistic regression, Lasso, group Lasso
are non-strongly convex. For such models, we show
that \cola enjoys  a  $\bigo{1/T}$ sublinear rate of convergence for all network topologies with a positive spectral gap. 

\begin{restatable}[Non-strongly Convex Case]{theorem}{TheoremSublinearRate}
	\label{theorem:sublinear_rate}
	Consider \cref{alg_dcocoa}, using a local solver of quality~$\Theta$. Let $g_i(\cdot)$ have $L$-bounded support, and let $f$ be $(1/\tau)$-smooth. 
	Let $\varepsilon_{G_\hh}>0$ be the desired duality gap. Then after $T$ iterations where
\begin{align*}
&\textstyle
T\ge
T_0 + \max \bigg\{ \left\lceil \frac{1}{\eta} \right\rceil,
\frac{4L^2\sigma\bar{\sigma}'}{\tau\varepsilon_{G_\hh}\eta}\bigg\} , \qquad
T_0 \ge
 t_0 + \bigg[\frac{2}{\eta} \left(\frac{8L^2\sigma\bar{\sigma}'}{\tau\varepsilon_{G_\hh}} -1 \right) \bigg]_+\\
&\textstyle
t_0\ge 
\max \left\{0, \left\lceil \frac{1 + \eta}{\eta } \log 
\frac{2\tau ( \ooa(\vc{\x}{0}, \{\vc{\v_l}{0}\}) -\ooa(\x^\star, \{\v^\star\}) ) }
{4L^2\sigma \bar{\sigma}'} \right\rceil\right\}
\end{align*}
and $\bar{\sigma}':=(1 + \beta)\sigma'$, $\alpha:=(1+\frac{(1-\beta)^2}{36(1+\Theta)\beta})^{-1}$ and $\eta:=\gamma(1-\Theta)(1-\alpha)$.
We have that the expected duality gap satisfies
\begin{align*}
\E \big[G_\hh(\bar{\x}, \{ \bar{\v}_k \}_{k=1}^K,\{\bar{\w}_k \}_{k=1}^K )\big]
\le \varepsilon_{G_\hh}
\end{align*}
at the averaged iterate $\bar{\x} := \frac{1}{T-T_0} \sum_{t=T_0+1}^{T-1} \vc{\x}{t}$,
and $\v_k':= \sum_{l=1}^K\mixingmat_{kl}\v_l$ and
$\bar{\v}_k := \frac{1}{T-T_0}\sum_{t=T_0+1}^{T-1} \vc{(\v_k')}{t}$
and $\bar{\w}_k := \frac{1}{T-T_0} \sum_{t=T_0+1}^{T-1} \nabla f(\vc{(\v_k')}{t})$.
\end{restatable}
Note that the assumption of bounded support for the $g_i$ functions is not restrictive in the general convex case, as discussed e.g. in \citep{Dunner:2016vga}.

\subsection{Local certificates for global accuracy}

Accuracy certificates for the training error are very useful for practitioners to diagnose the learning progress. In the centralized setting, the duality gap serves as such a certificate, and is available as a stopping criterion on the master node. In the decentralized setting of our interest, this is more challenging as consensus is not guaranteed. Nevertheless, we show in the following \cref{prop:localcertificate} that certificates for the decentralized objective \eqref{eq:H} can be computed from local quantities:
\begin{restatable}[Local Certificates]{proposition}{localcertificate}\label{prop:localcertificate}
	Assume $g_i$ has $L$-bounded support, 
	and let $\mathcal{N}_k:= \{j: \mixingmat_{jk} > 0\}$ be the set of nodes accessible to node $k$.
	Then for any given $\varepsilon > 0$, we have
	\[
	G_\mathcal{H} (\x; \{\v_k\}_{k=1}^K) \le \varepsilon ,
	\]
	if for all $k=1, \ldots, K$ the following two local conditions are satisfied:
	\begin{align}
			\v_k^\top \nabla f(\v_k) + 
		\sum_{i \in\Pk} \left( g_i(\x_i) + g_i^*(-\BA_i^\top\nabla f(\v_k)) \right) \ \le\ & \frac{\varepsilon}{2K} \label{eq:prop:14}
		\\
		\textstyle	\norm{\nabla f(\v_k) - \frac{1}{|\mathcal{N}_k|} \sum_{j\in \mathcal{N}_k}  \nabla f(\v_j)}_2
		\ \le\ & \textstyle
		 \left(\sum_{k=1}^{K} n_k^2 \sigma_{k} \right)^{-1/2} \frac{1-\beta}{2L\sqrt{K}} \varepsilon, \label{eq:prop:15}
	\end{align}
\end{restatable}

The local conditions \eqref{eq:prop:14} and \eqref{eq:prop:15} have a clear interpretation. 
	The first one ensures the duality gap of the local subproblem given by $\v_k$ as on the left hand side of \eqref{eq:prop:14} is small.
	The second condition \eqref{eq:prop:15} guarantees that consensus violation is bounded, by ensuring that the gradient of each node is similar to its neighborhood nodes. 
	
\begin{remark} The resulting certificate from \cref{prop:localcertificate} is local, in the sense that no global vector aggregations are needed to compute it.
For a certificate on the global objective, the boolean flag of each local condition \eqref{eq:prop:14} and \eqref{eq:prop:15} being satisfied or not needs to be shared with all nodes, but this requires extremely little communication.
Exact values of the  parameters $\beta$ and $\sum_{k=1}^{K} n_k^2 \sigma_{k}$  are not required to be known, and any valid upper bound can be used instead.
We can use the local certificates to avoid unnecessary work on local problems which are already optimized, as well as to continuously quantify how newly arriving local data has to be re-optimized in the case of online training.  The local certificates can also be used to quantify the contribution of newly joining or departing nodes, which is particularly useful in the elastic scenario described above. 
\end{remark}


\begin{figure}[b]
	
	\begin{subfigure}{.49\linewidth}
		\includegraphics[width=1.0\linewidth, keepaspectratio]{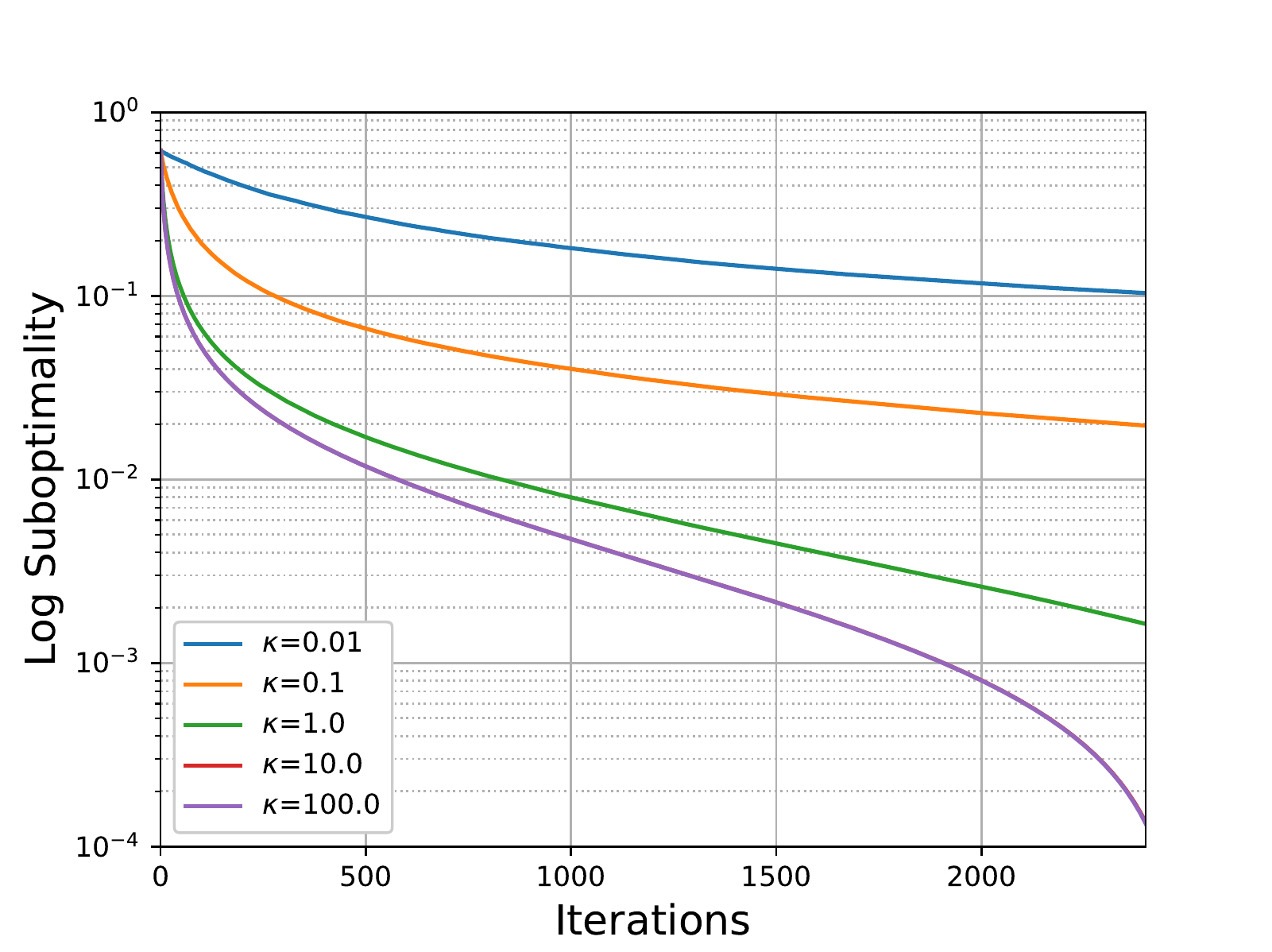}  
	\end{subfigure}
	\begin{subfigure}{.49\linewidth}
		\includegraphics[width=1.0\linewidth, keepaspectratio]{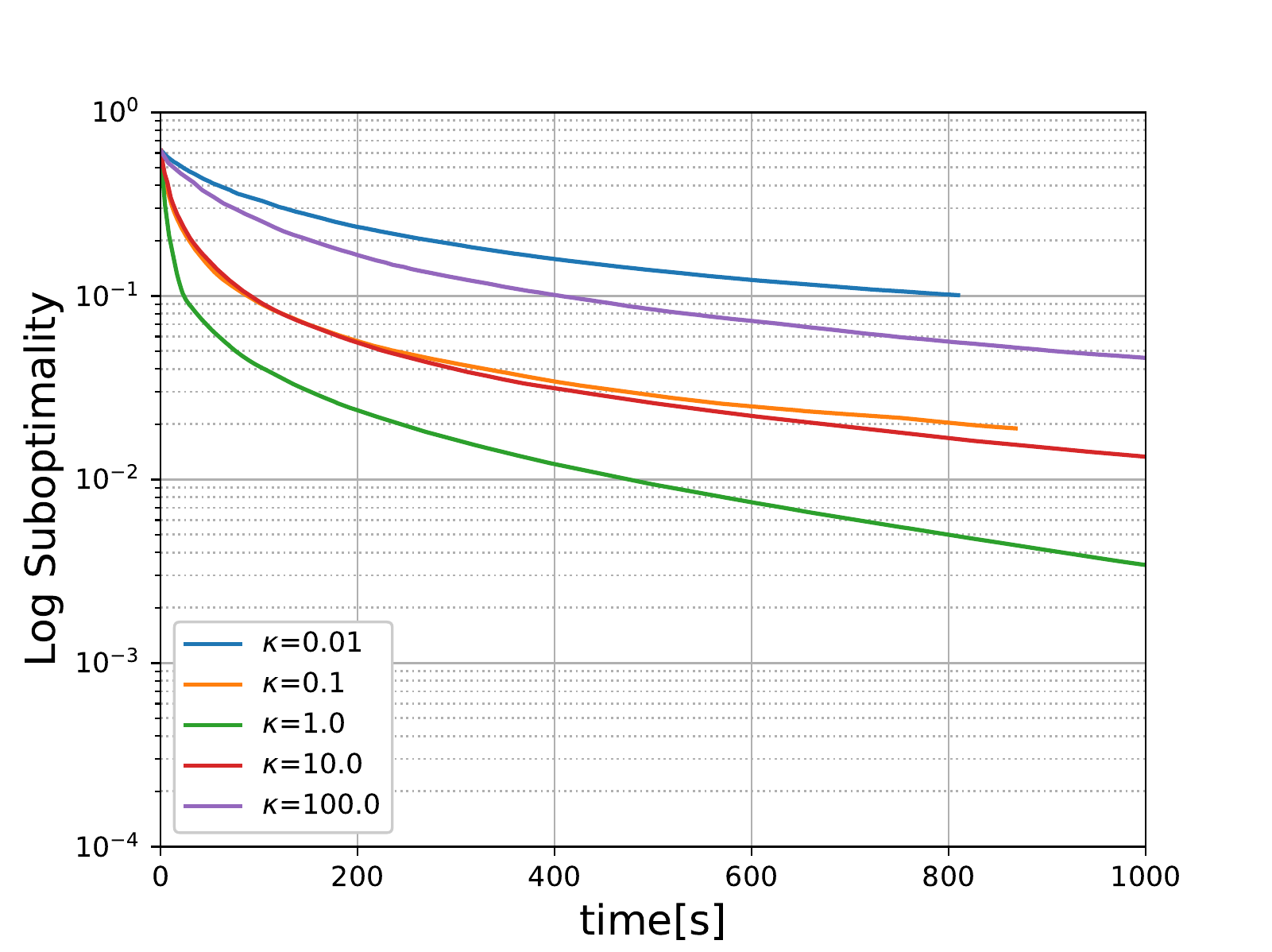}  
	\end{subfigure}
	\caption{
		Suboptimality for solving Lasso ($\lambda$=$10^{-6}$) for the RCV1 dataset 
		on a ring of 16 nodes.
		We illustrate the performance of \cola: a) number of iterations;
		b) time.  $\kappa$ here denotes the number of local data passes per communication round.
	} 
	\label{fig:varying_local_iters_plot}
\end{figure}

\section{Experimental results}

Here we illustrate the advantages of \cola in three respects: firstly we investigate the application in different network topologies and with varying subproblem quality~$\Theta$; 
secondly, we  compare \cola  with state-of-the-art decentralized baselines:
\textcircled{1}, {\diging} 
\citep{nedic2017achieving}, which generalizes the gradient-tracking technique of the  EXTRA algorithm \citep{shi2015extra}
, and 
\textcircled{2}, Decentralized ADMM (aka. consensus ADMM), which
extends the classical 
ADMM (Alternating Direction Method of Multipliers) method~\citep{boyd2011distributed} to the decentralized setting \citep{Shi:2014js,Wei:2013wy}; Finally, we show that \cola works 
in the challenging unreliable network environment where each node has a certain chance to drop out of the network.

We implement  all algorithms in PyTorch with MPI backend. The decentralized network topology is simulated by running one thread per graph node, on a $2$$\times$$12$ core Intel Xeon CPU E5-2680 v3 server with 256 GB RAM.
\cref{tab:datasets} describes the datasets\footnote{\url{https://www.csie.ntu.edu.tw/~cjlin/libsvmtools/datasets/}} used in the experiments. For Lasso, the columns of $\BA$ are features. For ridge regression, the columns are features and samples for \cola  primal and \cola dual, respectively. The order of columns is shuffled once before being distributed across the nodes.
Due to space limit, details on the experimental configurations are included in  \cref{sec:experiment_settings}.

\begin{wraptable}{r}{.55\linewidth}
	\caption{Datasets Used  for Empirical Study}
    \begin{tabular}{l|r|r|r} 
        \textbf{Dataset} & \!\!\textbf{\#Training} & \!\!\textbf{\#Features} & \!\!\textbf{Sparsity} \\\hline
        URL & 2M & 3M & 3.5e-5\\
        Webspam & 350K & 16M & 2.0e-4\\
		Epsilon & 400K& 2K& 1.0 \\
		RCV1 Binary & 677K & 47K & 1.6e-3
    \end{tabular}
    \vspace{2mm}
    \label{tab:datasets}
\end{wraptable}

\begin{figure}
\begin{subfigure}{.49\textwidth}
	\includegraphics[width=1.0\textwidth, keepaspectratio]{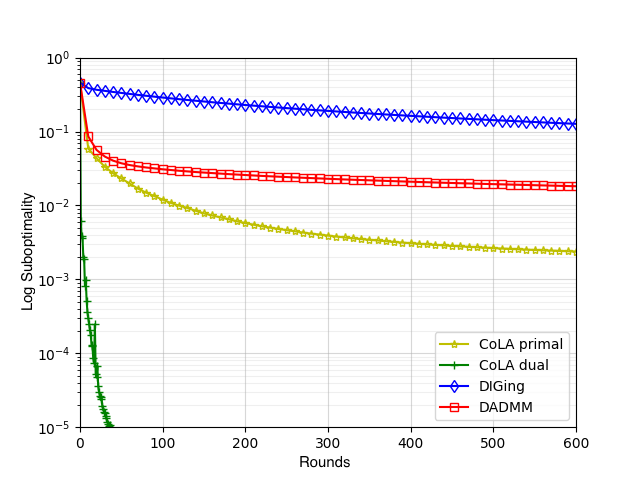}
	\label{fig:convergence:a}
\end{subfigure}
\begin{subfigure}{.49\textwidth}
\includegraphics[width=1.0\textwidth, keepaspectratio]{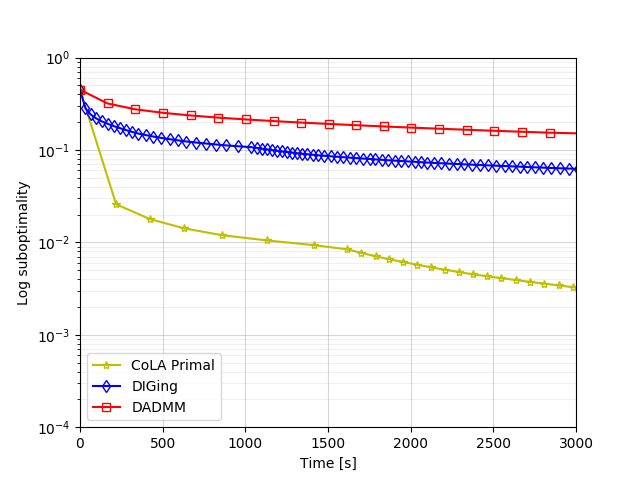}  
\label{fig:convergence:b}
\end{subfigure}
\vspace{0.cm}
\caption{Convergence of \cola for solving problems on a ring of $K$=$16$ nodes.
	Left) Ridge regression on  URL reputation dataset ($\lambda$=$10^{-4}$); Right) Lasso on  webspam dataset ($\lambda$=$10^{-5}$).
}
\label{fig:convergence}
\end{figure}

\textbf{Effect of approximation quality~$\Theta$. }
We study the convergence behavior in terms of the approximation quality $\Theta$.
Here, $\Theta$ is controlled by the number of data passes~$\kappa$ on subproblem~\eqref{eq_new_subproblem} per node.
\cref{fig:varying_local_iters_plot} shows that increasing $\kappa$ always results in less number of iterations (less communication rounds) for \cola. However, given a fixed 
network bandwidth, it leads to a clear trade-off
for the overall wall-clock time, showing the cost of both communication and computation. 
Larger~$\kappa$ leads to less communication rounds,
however, it also takes more time to solve subproblems. 
 The observations suggest  that one can adjust $\Theta$ for each node to handle system heterogeneity, as what we have discussed at the end of \cref{{sec:cola}}.

\textbf{Effect of graph topology. } 
Fixing $K$=$16$, we test the performance of \cola on 5 different topologies:  ring, 2-connected cycle, 3-connected cycle, 2D grid and complete graph.
The mixing matrix~$\mixingmat$ is given by Metropolis weights for all test cases (details in  \cref{sec:graphTopology}).
Convergence curves are plotted in \cref{fig:different_topologies}.
One can observe that for all  topologies, \cola converges monotonically and especailly when all nodes in the network are equal, smaller $\beta$ leads to a faster convergence rate. This is consistent with the intuition that $1-\beta$ measures the 
connectivity level of the topology.

\begin{figure}
   \begin{minipage}{0.49\textwidth}
   \centering
   \includegraphics[width=\textwidth, keepaspectratio]{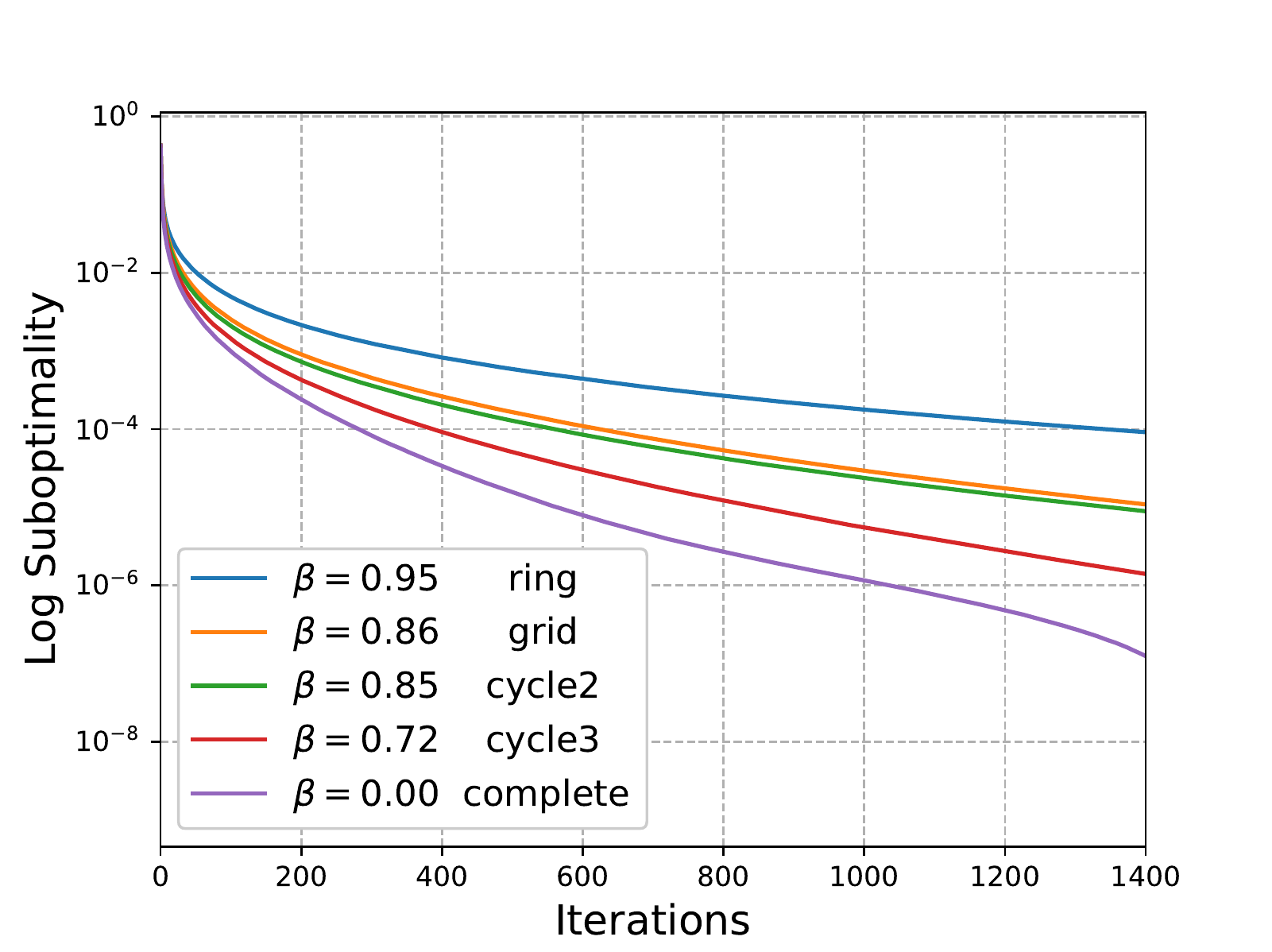}
   \caption{
   	Performance comparison of \cola on different topologies.
   	Solving Lasso regression ($\lambda$=$10^{-6}$) for RCV1 dataset with $16$ nodes.}
   \label{fig:different_topologies}
   \end{minipage}
   ~~
   \begin{minipage}{0.49\textwidth}
   \centering
   \includegraphics[width=\textwidth, keepaspectratio]{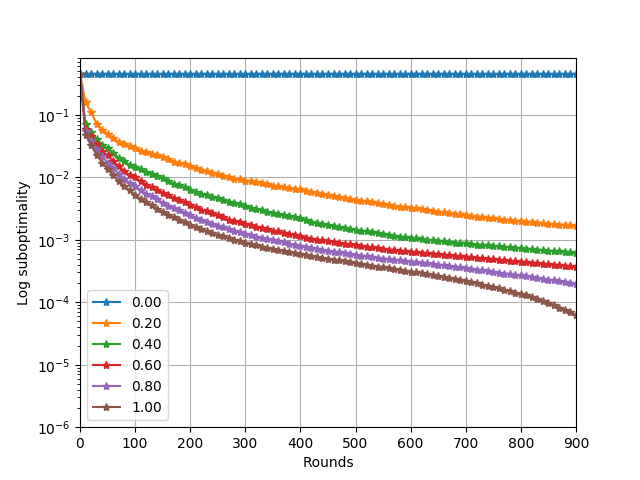}
   \caption{
   	Performance of \cola when nodes have $p$ chance of staying in the network on the URL dataset ($\lambda$=$10^{-4}$). Freezing $\vsubset{\x}{k}$ when node~$k$ leaves the network.
   }
   \label{fig:time_varying_graph}
   \end{minipage}
\end{figure}

\textbf{Superior performance compared to baselines. }
We compare \cola with \diging and D-ADMM for strongly and general convex problems.
For general convex objectives, we use Lasso regression with  $\lambda=10^{-4}$ on the webspam dataset;
for the strongly convex objective, we use Ridge regression with  $\lambda=10^{-5}$ on the  URL reputation dataset.
For Ridge regression, we can map \cola  to both primal and dual problems. 
 \cref{fig:convergence} traces the results on log-suboptimality. 
One can observe that for both generally and strongly convex objectives, \cola significantly outperforms \diging and decentralized ADMM in terms of number of communication rounds and computation time.
While \diging and D-ADMM need parameter tuning to ensure convergence and efficiency,
\cola is much easier to deploy as it is parameter free.
Additionally, convergence guarantees of ADMM relies on exact subproblem solvers, whereas inexact solver is allowed for \cola.

%
%

\textbf{Fault tolerance to unreliable nodes.}  
Assume each node of a network only has a chance of $p$ to participate in each round.
If a new node $k$ joins the network, then local variables are initialized as $\vsubset{\x}{k}=0$; if node~ $k$ leaves the network, then $\vsubset{\x}{k}$ will be frozen with $\Theta_k=1$.
All remaining nodes dynamically adjust their weights to maintain the doubly stochastic property of $\mixingmat$.
We run \cola on such unreliable networks of different $p$s and show the results in \cref{fig:time_varying_graph}.
First, one can observe that 
for all $p>0$ the suboptimality decreases monotonically
as \cola progresses.
It is also clear from the result that a smaller dropout rate (a larger $p$) leads to a faster convergence of \cola.


\section{Discussion and conclusions}

In this work we have studied training generalized linear models in the fully  decentralized setting.
%
We proposed a communication-efficient decentralized framework, termed 
\cola, which  is free of  parameter tuning. We proved that it has a sublinear rate of convergence for general convex problems, allowing e.g. L1 regularizers, and has a linear rate of convergence for strongly convex objectives. Our scheme offers primal-dual certificates which are useful in the decentralized setting. We demonstrated that \cola offers full adaptivity to heterogenous distributed systems on arbitrary network topologies, and is adaptive to changes in network size and data, and offers fault tolerance and elasticity. 
Future research directions include improving subproblems, as well as extension to the network topology with directed graphs, as well as recent communication compression schemes~\citep{stich2018sparsified}.
\\

\textbf{Acknowledgments.}
We thank Prof. Bharat K. Bhargava for fruitful discussions.
We acknowledge funding from SNSF grant 200021\_175796,  Microsoft Research JRC project `Coltrain', as well as a Google Focused Research Award.


{
\small
\bibliographystyle{unsrtnat} 
\bibliography{./bib}
}


\clearpage
\appendix
\appendixtitle{Appendix}

\section{Definitions}
\label{appe_def}
\begin{definition}[$L$-Lipschitz continuity]
	A function $h: \R^n \rightarrow \R$ is $L$-Lipschitz continuous
	if $\forall \;\u, \v\in \R^n$, it holds that 
		$$ |h(\u) - h(\v)| \leq L\| \u-\v \| .$$
\end{definition}
\begin{definition}[$1/\tau$-Smoothness]
		A differentiable function $f: \R^n \rightarrow \R$ is $1/\tau$-smooth if its gradient is $1/\tau$-Lipschitz continuous, or equivalently, $\forall~\u,\v$ it holds 
		\begin{equation}\textstyle
		f(\u) \leq f(\v) + \dtp{\nabla f(\v)}{\u-\v} + \frac{1}{2\tau} \|\u -\v\|^2.
		\end{equation}
\end{definition}
\begin{definition}[$L$-Bounded support]
	The function $g: \R^n \rightarrow \R \cup \{+\infty \}$ has $L$-bounded
	support if it holds $g(\u) < +\infty \Rightarrow \|\u\| \leq L$.
\end{definition}

\begin{definition}[$\mu$-Strong convexity] A function $h: \R^n \rightarrow \R$ is $\mu$-strongly convex for $\mu\geq 0$ if $\forall~\u, \v$, it holds $h(\u) \geq h(\v) + \dtp{\s}{\u - \v} + \frac{\mu}{2} \|\u - \v\|^2 $, for any $\s\in \partial h(\v)$, where   $\partial h(\v)$ is the subdifferential of $h$ at $\v$.
\end{definition}

\begin{lemma}[Duality between Lipschitzness and L-Bounded Support]  \label{lemma:duality}
	A generalization of \cite[Corollary 13.3.3]{rockafellar2015convex}. Given a proper convex function $g$ it holds that $g$ has $L$-bounded support w.r.t. the norm $\norm{.}$ if and only if $g^*$ is $L$-Lipschitz w.r.t. the dual norm $\norm{.}_*$.
\end{lemma}

 \section{Graph topology}\label{sec:graphTopology}
 Let $\mathcal{E}$ be the set of edges of a graph.
 For time-invariant undirected graph the mixing matrix should satisfy the following properties:
 \begin{enumerate}
 	\item (\textit{Double stochasticity}) $\mixingmat \one=\one$, $\one^\top\mixingmat=\one^\top$;
 	\item (\textit{Symmetrization}) For all $i, j$, $\mixingmat_{ij}=\mixingmat_{ji}$;
 	\item (\textit{Edge utilization}) If $(i, j)\in \mathcal{E}$, then
 	$\mixingmat_{ij}>0$; otherwise $\mixingmat_{ij}=0$.
 \end{enumerate}
 A desired mixing matrix can be constructed using Metropolis-Hastings weights \citep{hastings1970monte}:
 \begin{align*}
 \mixingmat_{ij} = 
 \begin{cases}
 1 / (1 + \max \{d_i, d_j \}), & \text{if } (i, j) \in \mathcal{E}\\
 0, & \text{if } (i, j) \not\in \mathcal{E} \text{ and } j\neq i\\
 1 - \sum_{l\in\mathcal{N}_i}\mixingmat_{il}, & \text{if }  j = i,
 \end{cases}
 \end{align*}
 where $d_i = | \mathcal{N}_i|$ is the degree of node  $i$.
 
 

\section{Proofs}\label{sec:proofs}
This section consists of three parts. Tools and observations are provided in \cref{ssec:prelemmas};
The main lemmas for the convergence analysis are proved in \cref{ssec:lemmas} ;
The main theorems and implications are proved in \cref{ssec:theorems}.

In some circumstances, it is convenient to use notations of array of stack column vectors. For example,
one can stack local estimates $\v_k$ to matrix $\BV := [\v_1; \cdots; \v_K]$, 
$\Delta \BV = [\Delta\v_1; \cdots; \Delta\v_K]$. 
The consensus vector $\v_c$ is repeated $K$ times which will be stacked similarly:
$\BV_c := \BA\x\one_K^\trans= \BV \BE$ where $\BE = \frac{1}{K} \one_K\one_K^\trans$.
The consensus violation under the two notations is written as
\begin{align*}\textstyle
\norm{\BV - \BV_c}_F^2 = \sum_{k=1}^{K} \norm{\v_k - \BA\x}_2^2.
\end{align*}
Then  Step~\labelcref{step_strategy2}  in \cola  is equivalent to 
\begin{equation}\label{eq:stacked_update_rule}
\vc{\BV}{t+1} = \vc{\BV}{t} \mixingmat + \gamma K \Delta \vc{\BV}{t}
\end{equation}
Besides, we also adopt following notations in the proof when there is no ambiguity: $\v_k' := \sum_{l=1}^{K} \mixingmat_{kl} \v_l$, 
$\locgra_k := \nabla f(\v_k)$, $\locgra_k' := \nabla f(\v_k')$
and $\bar{\locgra} := \frac{1}{K} \sum_{k=1}^{K} \locgra_k$.
For the decentralized duality gap $\colagap (\x, \{\v_k \}_{k=1}^K, \{\w_k \}_{k=1}^K )$, when $\w_k = \nabla f(\v_k)$, we simplify $\colagap (\x, \{\v_k \}_{k=1}^K, \{\w_k \}_{k=1}^K )$ to be $\colagap (\x, \{\v_k \}_{k=1}^K)$ in the sequel.

On a high level, we prove the convergence rates by bounding per-iteration reduction $\mathbb{E}[\vc{\ooa}{ t} - \vc{\ooa}{t+1}]$ using decentralized duality gap and other related terms,
then try to obtain the final rates by properly using specific properties of the objectives.

However,  the specific  analysis of the new fully decentralized algorithm \cola poses many new  challenges, and we propose significantly new proof techniques in the analysis. Specifically, i) we introduce the 
decentralized duality gap, which is suited for the   decentralized algorithm \cola; ii) 
consensus violation is the usually   challenging part  in analyzing decentralized algorithms. 
Unlike using uniform bounds for consensus violations, e.g., \citep{yuan2016convergence}, we properly combine  
the  consensus violation term and the objective decrease term  (c.f. \cref{lemma:sum_of_subproblems_reductionNewGH_CV,lemma:sum_of_subproblems_reductionNewGH_timeinvar_prime}),  thus reaching  arguably 
tight convergence bounds   for both the consensus violation term and the objective.

\subsection{Observations and properties}\label{ssec:prelemmas}
In this subsection we introduce basic lemmas. 
\cref{lemma:DH} establishes the relation between $\{\vv_k\}_{k=1}^K$ and ${\v_c}$ and  bounds     $\OA(\xv)$ using  $\ooa (\xv)$ and the  consensus violation. 


\restatlemmaDalphaHalpha*
\begin{proof}[Proof of \cref{lemma:DH}] 

Let ${\widetilde{\vv}} := \frac{1}{K} \sum_{k=1}^{K} \v_k $. 
Using the doubly stochastic property of the matrix $\mixingmat$
\begin{align*}
\notag
\vc{\widetilde{\vv}}{t+1}
& =\frac{1}{K}\sum_{k=1}^K\vc{\vv_k}{t+1}
=\frac{1}{K}\sum_{k=1}^K \left( \sum_{l=1}^K \mixingmat_{kl} \vc{\vv}{t}_l  + \gamma K\Delta \vc{\vv_k}{t}  \right)\\
&=\frac{1}{K}\sum_{l=1}^K \vc{\vv}{t}_l  + {\gamma}\sum_{k=1}^K\Delta \vc{\vv_k}{t}                       
=\vc{\widetilde{\vv}}{t}  + \gamma\sum_{k=1}^K \Delta \vc{\vv_k}{t}
\end{align*}
On the other hand, $\vv_c^{(t)}:=\BA\vc{\x}{t}$ is updated based on all changes of local variables $\{\vsubset{\xv}{k}\}_{k=1}^K$
$$ \vv_c^{(t+1)} = \vv_c^{(t)} + \gamma \sum_{k=1}^K\Delta \vc{\vv_k}{t}.$$
Since $\vc{\widetilde{\vv}}{0} = \vv_c^{(0)}$, we can conclude that $\vc{\widetilde{\vv}}{t} = \vv_c^{(t)}$ $\forall~t$.	
From convexity of $f$ we know
\begin{align*}
\OA(\xv)
=f(\vv_c) + g(\xv)
= f\left(\frac{1}{K}\sum_{k=1}^{K} \vv_k\right)+ g(\xv)
\le \frac{1}{K}\sum_{k=1}^{K} f(\vv_k)+ g(\xv)
=\mathcal{H}(\xv)
\end{align*}
Using $1/\tau$-smoothness of $f$ gives
\begin{align*}
\ooa (\xv)
=&\frac{1}{K}\sum_{k=1}^{K} f(\vv_k)+ g(\xv)\\
\le&
\frac{1}{K}\sum_{k=1}^{K}
\left(
f(\vv_c)
+
\nabla f(\vv_c) ^\top ( \vv_k - \vv_c )
+
\frac{1}{2\tau} \norm{\vv_k- \vv_c}^2
\right)
+ g(\xv)\\
= &
\OA(\xv) +
\frac{1}{2\tau K}\sum_{k=1}^{K} \norm{\vv_k- \vv_c}^2. 
\end{align*}
\end{proof}

The following lemma introduces the dual problem and the duality gap of \eqref{eq:H}.
\restaDecenDual* 

\begin{proof}
	Let $\blambda_k$ be the Lagrangian multiplier for the constraint $\vv_k = \BA \x$,  the 	Lagrangian  function   is
	\begin{align*}
	L(\x, \{\v_k \}_{k=1}^{K}, \{\blambda_k \}_{k=1}^{K} )
	= \frac{1}{K} \sum_{k=1}^{K} f(\v_k) + \sum_{i =1}^n g_i(\x_i) + \sum_{k=1}^{K} \langle \blambda_k , \BA\x - \v_k \rangle
	\end{align*}
	The dual problem of \eqref{eq:H} follows by taking the infimum with  respect to both $\x$ and $\{\v_k \}_{k=1}^{K}$:
	\begin{align*}
	&\inf_{ \x, \{\v_k \}_{k=1}^K} 	L(\x, \{\v_k \}_{k=1}^{K}, \{\blambda_k \}_{k=1}^{K} )\\
	=&
	\inf_{ \x, \{\v_k \}_{k=1}^K} 
	\frac{1}{K} \sum_{k=1}^{K} f(\v_k) + \sum_{i =1}^n g_i(\x_i) + \sum_{k=1}^{K} \langle \blambda_k , \BA\x - \v_k \rangle
	\\
	=&
	\sum_{k=1}^{K} \inf_{ \{\v_k \}_{k=1}^K} 
	( \frac{1}{K} f(\v_k) -  \langle \blambda_k , \v_k \rangle 
	)
	+
	\inf_{ \x} (\sum_{i =1}^n g_i(\x_i) + \sum_{k=1}^{K} \langle \blambda_k , \BA\x\rangle 
	) \\
	=&
	- \sum_{k=1}^{K} \sup_{ \{\v_k \}_{k=1}^K} 
	( \langle \blambda_k , \v_k \rangle  - \frac{1}{K} f(\v_k)
	)
	- \sup_{ \x} ( -\sum_{k=1}^{K} \langle \a_k , \BA\x\rangle - \sum_{i =1}^n g_i(\x_i)
	) \\
	=&
	- \sum_{k=1}^{K} \frac{1}{K} f^*(K \blambda_k)
	- \sum_{i =1}^n g^*_i ( -\sum_{k=1}^{K}  \BA_i^\top \blambda_k  )
	\end{align*}
	
	Let us change variables from $\blambda_k$ to $\w_k$ by setting
	$\w_k := K\blambda_k$. 
	If written in terms of minimization, the Lagrangian  dual of $\ooa$  is
	\begin{align}\textstyle
	\min_{ \{\w_k \}_{k=1}^K}  \oob {\{\w_k \}_{k=1}^K}
	= \frac{1}{K} \sum_{k=1}^{K}  f^*(\w_k) 
	+ \sum_{i =1}^n g^*_i \left( - \frac{1}{K}\sum_{k=1}^{K}  \BA_i^\top \w_k \right)
	\end{align}
	The optimality condition is that $\w_k = \nabla f(\v_k)$.
	Now we can see the duality gap  is
	\begin{align*}
	\colagap  =&   \ooa + \oob \\
	=& 
	\frac{1}{K} \sum_{k=1}^{K} f(\v_k) + \sum_{i=1}^{n} g_i(x_i) + 
 \frac{1}{K} \sum_{k=1}^{K}  f^*(\w_k) 
+ \sum_{i =1}^n g^*_i \left( - \frac{1}{K}\sum_{k=1}^{K}  \BA_i^\top \w_k  \right)
	\end{align*}
\end{proof}
The following lemma correlates the consensus violation with the magnitude of the $\v$ parameter updates~$\norm{\Delta \v_k}_2^2$.
\begin{lemma}\label{lemma:consensus_violation}
	The consensus violation during the execution of \cref{alg_dcocoa} can be bound by
	\begin{align}\label{eq:consensus_violation}
	\sum_{k=1}^{K} \norm{\vc{\v_k}{t+1} - \vc{\BA\x}{t+1}}_2^2 \le 
	\beta \sum_{k=1}^{K} \norm{\vc{\v_k}{t} - \vc{\BA\x}{t}}_2^2
	+ (1-\beta)c_1(\beta, \gamma, K)  \sum_{k=1}^{K} \norm{\Delta \vc{\v_k}{t}}_2^2
	\end{align}
	where $c_1(\beta, \gamma, K):= \gamma^2K^2/(1- \beta)^2$.
\end{lemma}
\begin{proof}
	Consider the norm of consensus violation at time $t+1$ and apply Algo. Step~\labelcref{step_strategy2}
	\begin{equation*}\small
	\norm{\vc{\mathbf{V}}{t+1} - \vc{\mathbf{V}}{t+1}_c}_F^2
	=
	\norm{\vc{\mathbf{V}}{t+1} (\BI - \BE)}_F^2
	=
	\norm{(\vc{\mathbf{V}}{t}\mixingmat
		+
		\gamma K \vc{\Delta \mathbf{V}}{t})
		(\BI - \BE)}_F^2.
	\end{equation*}
	Further, use $\mixingmat(\BI - \BE) = (\BI - \BE)(\mixingmat- \BE)$, $\norm{\BI - \BE}_\infty=1$,
	and Young's inequality with $\varepsilon_\v$
	\begin{align*}
	\norm{\vc{\mathbf{V}}{t+1} - \vc{\mathbf{V}}{t+1}_c}_F^2\le
	(1 + \varepsilon_\v) \norm{\vc{\mathbf{V}}{t} (\BI - \BE) (\mixingmat - \BE)}_F^2
	+
	(1 + \frac{1}{\varepsilon_\v})
	\gamma^2 K^2\norm{\vc{\Delta \mathbf{V}}{t}}_F^2.
	\end{align*}
	Use the spectral property of $\mixingmat$ we therefore have:
	\begin{align}\label{eq:consensus_violation:1}
	\norm{\vc{\mathbf{V}}{t+1} - \vc{\mathbf{V}}{t+1}_c}_F^2\le
	(1 + \varepsilon_\v) \beta^2
	\norm{\vc{\mathbf{V}}{t} - \vc{\BV}{t}_c}_F^2
	+ (1 + \frac{1}{\varepsilon_\v}) \gamma^2 K^2 \norm{\Delta \vc{\mathbf{V}}{t}}_F^2.
	\end{align}
	Recursively apply \eqref{eq:consensus_violation:1} for $i=0, \ldots, t-1$ gives
	\begin{align}\label{eq:consensus_violation:2}
	\norm{\vc{\mathbf{V}}{t} - \vc{\mathbf{V}}{t}_c}_F^2 \le
	(1 + \frac{1}{\varepsilon_\v}) \gamma^2 K^2 
	\sum_{i=0}^{t-1} ((1 + \varepsilon_\v) \beta^2)^{t-1-i} 
	\norm{\Delta \vc{\mathbf{V}}{i}}_F^2.
	\end{align}
	Consider $\norm{\Delta \vc{\mathbf{V}}{t}}_F^2$ generated at time $t$, it will be used in \eqref{eq:consensus_violation:2} from time $t+1, t+2, \ldots,$ with coefficients
	$1, (1+\varepsilon_{\v})\beta^2, ((1+\varepsilon_{\v})\beta^2)^2, \ldots$. Sum of such coefficients are finite
	\begin{align}\label{eq:consensus_violation:3}
	(1 + \frac{1}{\varepsilon_\v}) \gamma^2 K^2 
	\sum_{t=T}^{\infty} ((1 + \varepsilon_\v) \beta^2)^{t-T}
	\le 
	\gamma^2 K^2 \frac{1 + 1/\varepsilon_\v}{1- (1 + \varepsilon_\v) \beta^2}
	=: c_1(\beta, \gamma, K)
	\end{align}
	where we need $(1 + \varepsilon_\v) \beta^2 < 1$. To minimize $c_1(\beta, \gamma, K)$ we can choose $\varepsilon_\v=1/\beta - 1$
	\begin{align}
	c_1(\beta, \gamma, K) = \gamma^2 K^2 / (1 - \beta)^2
	\end{align}
	
	Then \eqref{eq:consensus_violation:1} becomes
	\begin{align*}
	\sum_{k=1}^{K} \norm{\vc{\v_k}{t+1} - \vc{\BA\x}{t+1}}_2^2 \le 
	\beta \sum_{k=1}^{K} \norm{\vc{\v_k}{t} - \vc{\BA\x}{t}}_2^2
	+ (1-\beta)c_1(\beta, \gamma, K)  \sum_{k=1}^{K} \norm{\Delta \vc{\v_k}{t}}_2^2
	\end{align*}
\end{proof}

\begin{lemma}\label{lemma:convexity:subproblem}
	Let $\Delta{\vsubset{\xv^\star}{k}}$ and $\Delta{\vsubset{\xv}{k}}$ be the exact and $\Theta$-inexact solution of subproblem $\mathscr{G}^{\sigma'}_k\hspace{-0.08em}({}\cdot{}; \v_k, \vsubset{\xv}{k})$.
	The change of iterates satisfies the following inequality
	\begin{equation}
		\frac{\sigma'}{4\tau} \sum_{k=1}^K \norm{\BA\vsubset{\Delta \x}{k}}_2^2
		\le (1+\Theta)(\ooa(\0; \{\v_k\}) - \ooa(\Delta\x^\star; \{\v_k\}))
	\end{equation}
\end{lemma}
\begin{proof}
	First use the Taylor expansion of $\mathscr{G}^{\sigma'}_k\hspace{-0.08em}({}\cdot{}; \v_k, \vsubset{\xv}{k})$ and the defnition of $\Delta{\vsubset{\xv^\star}{k}}$ we have
	\begin{equation}\label{eq:convexity:subproblem}
		\frac{\sigma'}{2\tau} \norm{\BA\vsubset{(\Delta\z - \Delta \x^\star)}{k}}_2^2 \le
		\mathscr{G}^{\sigma'}_k\hspace{-0.08em}\left(\Delta\vsubset{\z}{k}; \v_k, \vsubset{\xv}{k}\right)
		- \mathscr{G}^{\sigma'}_k\hspace{-0.08em}\left(\Delta{\vsubset{\xv^\star}{k}}; \v_k, \vsubset{\xv}{k}\right)
	\end{equation}
	for all $\Delta\vsubset{\z}{k} \in \R^n$ and $k=1,\ldots,K$.
	Apply \eqref{eq:convexity:subproblem} with $\Delta\vsubset{\z}{k}=\0$ for all $k$ and sum them up yields
	\begin{equation}\label{eq:convexity:subproblem:2}
		\frac{\sigma'}{2\tau}\sum_{k=1}^K \norm{\BA\vsubset{\Delta \x^\star}{k}}_2^2
		\le \ooa(\0; \{\v_k\}) - \ooa(\Delta\x^\star; \{\v_k\})
	\end{equation}
	Similarly, apply \eqref{eq:convexity:subproblem} for $\Delta\vsubset{\z}{k} = \Delta\vsubset{\x}{k}$ for all $k$ and sum them up gives
	\begin{equation}
		\frac{\sigma'}{2\tau}\sum_{k=1}^K \norm{\BA
		\vsubset{(\Delta\x - \Delta \x^\star)}{k}}_2^2
		\le \ooa(\Delta\x; \{\v_k\}) - \ooa(\Delta\x^\star; \{\v_k\})
	\end{equation}
	By \cref{assumption:theta} the previous inequality becomes
	\begin{equation}\label{eq:convexity:subproblem:3}
		\frac{\sigma'}{2\tau}\sum_{k=1}^K \norm{\BA
		\vsubset{(\Delta\x - \Delta \x^\star)}{k}}_2^2
		\le \Theta (\ooa(\0; \{\v_k\}) - \ooa(\Delta\x^\star; \{\v_k\}))
	\end{equation}
	%
	%
	The following inequality is straightforward
	\begin{equation}\label{eq:convexity:subproblem:4}
		\frac{1}{2} \sum_{k=1}^K \norm{\BA\vsubset{\Delta \x}{k}}_2^2
		\le \sum_{k=1}^K \norm{\BA\vsubset{\Delta \x^\star}{k}}_2^2
		+ \sum_{k=1}^K
		\norm{\BA \vsubset{(\Delta\x - \Delta \x^\star)}{k}}_2^2
	\end{equation}
	Multiply \eqref{eq:convexity:subproblem:4} with $\sigma'/(2\tau)$ and use \eqref{eq:convexity:subproblem:2} and \eqref{eq:convexity:subproblem:3}
	\begin{equation}
		\frac{\sigma'}{4\tau} \sum_{k=1}^K \norm{\BA\vsubset{\Delta \x}{k}}_2^2
		\le (1+\Theta)(\ooa(\0; \{\v_k\}) - \ooa(\Delta\x^\star; \{\v_k\}))
	\end{equation}
\end{proof}

\subsection{Main lemmas}\label{ssec:lemmas}
We first present two main lemmas for the per-iteration improvement.

\begin{lemma}
	\label{lemma:sum_of_subproblems_reductionNewGH_CV}
	Let $g_i$ be strongly convex with convexity parameter $\mu_g\ge 0$ with respect to the norm $\norm{\cdot}, \forall~i\in[n]$.
	Then for all iterations $t$ of outer loop, and any $s\in[0, 1]$, it holds that
	\begin{equation}\label{eq:sum_of_subproblems_reductionNewGH_CV}
	\begin{split}
	&\mathbb{E} \left[\ooa(\vc{\x}{t}; \vc{\v_k}{t}) - \ooa(\vc{\x}{t+1}; \vc{\v_k}{t+1})
		-\alpha\frac{\gamma\sigma_1'}{2\tau} \sum_{k=1}^K \norm{\BA \Delta \vc{\vsubset{\x}{k}}{t} }_2^2
	\right]\\
	\ge&
	\eta \left(
	sG_\mathcal{H}(\vc{\xv}{t}; \{{\vc{\v_k'}{t}}\}_{k=1}^K)
	-  \frac{s^2\bar{\sigma}'}{2\tau} \vc{R}{t}
	\right)
	- \frac{9\beta\eta}{ 2\tau\sigma'} \sum_{k=1}^K \norm{\vc{\v_k}{t} - \vc{\BA\x}{t} }_2^2
	\end{split}
	\end{equation}
	where $\alpha\in[0,1]$ is a constant and $\eta:=\gamma(1-\Theta)(1-\alpha)$ and $\sigma_1' := \frac{(1-\Theta)}{2(1+\Theta)}\sigma'$
	and $\bar{\sigma}':= (1 + \beta)\sigma'$ and $\v_k':= \sum_{l=1}^K \mixingmat_{kl} \v_l$.
	\begin{equation}
	\textstyle
	R^{(t)}:= - \frac{\tau\mu_g  (1-s)}{\bar{\sigma}'s}\norm{\vc{\uv}{t} - \vc{\xv}{t} }^2
	+ \sum_{k=1}^{K}\norm{\BA\vsubset{(\vc{\uv}{t} - \vc{\xv}{t})}{k}}^2
	\end{equation}
	for $\vc{\uv}{t}\in\R^n$ with 
	$\vc{\bar{\locgra}'}{t}:= \frac{1}{K}\sum_{k=1}^K \nabla f(\vc{\v_k'}{t})$
	\begin{equation}
	\vc{u_{i}}{t} \in \partial g_i^*(-\BA_i^\top \vc{\bar{\locgra}'}{t} ) \qquad k \text{ s.t. }i \in \Pk
	\end{equation}
\end{lemma}
\begin{proof}[Proof of \cref{lemma:sum_of_subproblems_reductionNewGH_CV}]
	

	For simplicity, we write $\vc{\ooa}{t}$ instead of $\ooa(\vc{\x}{t}; \{\vc{\v_k}{t}\}_{k=1}^K)$ and $\v_k':= \sum_{i=1}^K\mixingmat_{ik}\v_i$.
	\begin{align*}
	&\mathbb{E}[\vc{\ooa}{t} - \vc{\ooa}{t+1}]\\
	=& \frac{1}{K} \sum_{k=1}^K
	f({{\vv_k}})
	-
	\frac{1}{K} \sum_{k=1}^K
	f\left(\v_k'
	+ \gamma K \Delta \vsubset{\vv}{k}
	\right)
	+ g(\xv) - g(\xv + \gamma \Delta \xv)\\
	=&
	\underbrace{
		\frac{1}{K} \sum_{k=1}^K
		f({{\vv_k}})
		-
		\frac{1}{K} \sum_{k=1}^K
		f\left(\v_k'\right)
	}_{D_1}\\
	&+
	\underbrace{
		\sum_{k=1}^K
		\left\{
		\frac{1}{K}
		f\left(\v_k'\right)
		+
		g(\vsubset{\xv}{k})
		\right\}
		-
		\sum_{k=1}^K
		\left\{
		\frac{1}{K}
		f\left(\v_k' + \gamma K \Delta\vsubset{\vv}{k}\right)
		+
		g(\vsubset{\xv}{k} + \gamma\vsubset{\Delta\xv}{k})
		\right\}
	}_{D_2}
	\end{align*}
	By the convexity of $f$, $D_1\ge 0$.
	Using the convexity of $f$ and $g$ in $D_2$ we have
	\begin{align*}
	\frac{1}{\gamma} D_2 \ge &
	\sum_{k=1}^K
	\left\{
	\frac{1}{K}
	f\left(\v_k'\right)
	+
	g(\vsubset{\xv}{k})
	\right\}
	-
	\sum_{k=1}^K
	\left\{
	\frac{1}{K}
	f\left(\v_k' + K \Delta\vsubset{\vv}{k}\right)
	+
	g(\vsubset{\xv}{k} +\vsubset{\Delta\xv}{k})
	\right\}\\
	\ge &\mathbb{E}\left[
	\sum_{k=1}^{K} \newsub\left(\0; \v_k', \vsubset{\xv}{k}\right)
	- \sum_{k=1}^{K} \newsub\left(\Delta{\vsubset{\xv}{k}}; \v_k', \vsubset{\xv}{k}\right)
	\right]
	\end{align*}
	%
	%
	Use the \cref{assumption:theta} we have
	\begin{align*}
	D_2 \ge & \gamma \mathbb{E}\left[
	\sum_{k=1}^{K} \mathscr{G}^{\sigma'}_k\hspace{-0.08em}\left(\0; \v_k', \vsubset{\xv}{k}\right)
	- \sum_{k=1}^{K} \mathscr{G}^{\sigma'}_k\hspace{-0.08em}\left(\Delta{\vsubset{\xv}{k}}; \v_k', \vsubset{\xv}{k}\right)\right] \\
	\ge &
	\gamma (1-\Theta)
	\underbrace{
		\sum_{k=1}^{K}
		\left\{
		\mathscr{G}^{\sigma_2'}_k\hspace{-0.08em}\left(\0; \v_k', \vsubset{\xv}{k}\right)
		-  \mathscr{G}^{\sigma_2'}_k\hspace{-0.08em}\left(\Delta{\vsubset{\xv}{k}^\star}; \v_k', \vsubset{\xv}{k}\right)
		\right\}
	}_{C}
	\end{align*}
	Let $\alpha\in[0, 1]$ and apply \cref{lemma:convexity:subproblem}, the previous inequality becomes
	\begin{equation}
		D_2 \ge \gamma (1-\Theta) (1-\alpha) C
		+ \alpha\frac{\gamma\sigma_1'}{2\tau} 
			\sum_{k=1}^K \norm{\BA \Delta \vsubset{\x}{k}}_2^2
	\end{equation}
	where $\sigma_1' := \frac{(1-\Theta)}{2(1+\Theta)}\sigma'$.
	From the definition of $u_i$ we know
	\begin{equation}\label{eq:dual:u}
	g_i(u_i) = u_i(-\BA_i^\top \bar{\locgra}') - g_i^*(-\BA_i^\top \bar{\locgra}')
	\end{equation}
	Replacing $\Delta \xv_i=s(u_i - x_i)$ in $C$ gives
	\begin{align*}
	C \ge& \sum_{k=1}^{K} \left\{
	\sum_{i \in \Pk} (g_i(x_i) - g_i(x_i + \Delta x_i))
	-
	\left\langle
	\locgra_k',
	\BA \Delta \vsubset{\xv}{k}
	\right\rangle
	- \frac{\sigma'}{2\tau} \norm{\BA \Delta \vsubset{\xv}{k}}^2
	\right\}\\
	\ge& \sum_{k=1}^{K} \left\{
	\sum_{i \in \Pk} (sg_i(x_i) - sg_i(u_i) + \frac{\mu_g}{2\tau}(1-s)s(u_i - x_i)^2)
	-
	\left\langle
	\locgra_k',
	\BA \Delta \vsubset{\xv}{k}
	\right\rangle
	- \frac{\sigma'}{2\tau} \norm{\BA \Delta \vsubset{\xv}{k}}^2
	\right\}\\
	\stackrel{\eqref{eq:dual:u}}{=} &
	\sum_{k=1}^K (\sum_{i\in\Pk} \left(sg_i(x_i)
	+ sg^*_i(-\BA_i^\top \bar{\locgra}'))
	+ s\langle \v_k/K, \locgra_k'\rangle
	\right)
	-
	s\sum_{k=1}^K (\langle \v_k/K, \locgra_k'\rangle - \langle \BA\vsubset{\u}{k}, \bar{\locgra}'\rangle)\\
	&+
	\sum_{k=1}^K\sum_{i\in\Pk}\left\{
	\frac{\mu_g}{2\tau}(1-s)s(u_i - x_i)^2
	\right\}
	-
	\sum_{k=1}^{K}
	\left\langle
	\locgra_k',
	\BA \Delta \vsubset{\xv}{k}
	\right\rangle
	- \sum_{k=1}^{K} \frac{\sigma'}{2\tau} \norm{\BA \Delta \vsubset{\xv}{k}}^2\\
	=& \sum_{k=1}^K \big \{\sum_{i\in\Pk} \left(sg_i(x_i)
	+ sg^*_i(-\BA_i^\top \bar{\locgra}')\right)
	+ s \langle \v_k/ K,  {\locgra}_k'\rangle
	\big \}
	+ \frac{\mu_g }{2}(1-s)s\norm{\uv-\xv}^2\\
	&- \sum_{k=1}^{K} \frac{s^2\sigma'}{2\tau} \norm{\BA (\uv-\xv)_{[k]}}^2
	-
	s\sum_{k=1}^K 
	(\langle \v_k/K, \locgra_k'\rangle
	- \langle \BA\vsubset{\u}{k}, \bar{\locgra}'\rangle
	+
	\langle \locgra_k', \BA\vsubset{(\u - \x)}{k} \rangle)\\
	=&
	sG_\mathcal{H}(\xv; \{{\v_k'}\}_{k=1}^K)
	+ \frac{\mu_g}{2}(1-s)s\norm{\uv-\xv}^2\\
	&-  \frac{s^2\sigma'}{2\tau} \sum_{k=1}^{K}\norm{\BA (\uv-\xv)_{[k]}}^2
	-
	s\sum_{k=1}^K 
	(\langle \v_k'/K, \locgra_k'\rangle
	- \langle \BA\vsubset{\u}{k}, \bar{\locgra}'\rangle
	+
	\langle \locgra_k', \BA\vsubset{(\u - \x)}{k} \rangle)
	\end{align*}
	We can bound the last term of the previous equation as $D_3$
	\begin{align*}
	\frac{1}{s}
	D_3=&\sum_{k=1}^K 
	(\langle \v_k'/K, \locgra_k'\rangle
	- \langle \BA\vsubset{\u}{k}, \bar{\locgra}'\rangle
	+
	\langle \locgra_k', \BA\vsubset{(\u - \x)}{k} \rangle
	)\\
	=&
	\sum_{k=1}^K 
	(\langle \locgra_k', \v_k'/K\rangle
	- \langle \locgra_k', \BA\u/K \rangle
	+
	\langle \locgra_k', \BA\vsubset{(\u - \x)}{k} \rangle
	)\\
	=&
	\frac{1}{K} \sum_{k=1}^K 
	\langle \locgra_k', \v_k' - \BA\x\rangle
	- \sum_{k=1}^K 
	\langle \bar{\locgra}', \BA\vsubset{(\u - \x) }{k} \rangle
	+ \sum_{k=1}^K 
	\langle \locgra_k', \BA\vsubset{(\u - \x)}{k} \rangle\\
	=&
	\frac{1}{K} \sum_{k=1}^K 
	\langle \locgra_k' - \bar{\locgra}', \v_k' - \BA\x\rangle
	+ \sum_{k=1}^K 
	\langle \locgra_k' - \bar{\locgra}', \BA\vsubset{(\u - \x)}{k} \rangle
	\end{align*}
	
	Bound the gradient terms with consensus violation.
	First bound $\sum_{k=1}^{K} \norm{\locgra_k' - \bar{\locgra}'}_2^2$, 
	define $\locgra_c := \nabla f(\BA\x)$
	\begin{align*}
	\sum_{k=1}^{K} \norm{\locgra_k' - \bar{\locgra}'}_2^2
	\le& 2 \sum_{k=1}^{K} \left(
	\norm{\locgra_k' - \locgra_c}_2^2 + 
	\norm{\locgra_c - \bar{\locgra}'}_2^2
	\right)
	\le 2 \sum_{k=1}^{K}
	\norm{\locgra_k' - \locgra_c}_2^2 + 
	2 \frac{1}{K} \sum_{k=1}^{K}
	\norm{\locgra_c - {\locgra}_k'}_2^2
	\end{align*}
	Apply the $1/\tau$-smoothness of $f$ we have
	\begin{equation}\label{eq:sum_of_subproblems_reductionNewGH_CV:1}
	\sum_{k=1}^{K} \norm{\locgra_k' - \bar{\locgra}'}_2^2
	\le \frac{4}{\tau^2} \sum_{k=1}^{K}
	\norm{ \v_k' - \BA\x }_2^2
	\le \frac{4\beta^2}{\tau^2} \sum_{k=1}^{K}
	\norm{ \v_k - \BA\x }_2^2
	\end{equation}

	Bound the first term in $D_3$
	\begin{align*}
	s\frac{1}{K} \sum_{k=1}^K 
	\langle \locgra_k' - \bar{\locgra}', \v_k' - \BA\x\rangle
	\le& \frac{s}{2K}\sum_{k=1}^{K} 
	\left( \tau \norm{\locgra_k' - \bar{\locgra}'}_2^2
	+ \frac{1}{\tau}\norm{ \v_k' - \BA\x }_2^2
	\right) \\
	\stackrel{\eqref{eq:sum_of_subproblems_reductionNewGH_CV:1}}{\le}&
	\frac{5\beta^2s}{2\tau K} \sum_{k=1}^{K}  \norm{ \v_k - \BA\x }_2^2
	\end{align*}
	
	Bound the second term in $D_3$
	\begin{align*}
	&s \sum_{k=1}^K \langle \locgra_k' - \bar{\locgra}', \BA\vsubset{(\u - \x)}{k} \rangle
	\le
	\frac{\tau}{2\sigma'\beta} \sum_{k=1}^K \norm{\locgra_k' - \bar{\locgra}' }_2^2
	+ \frac{s^2\sigma'\beta}{2\tau} \sum_{k=1}^K \norm{\BA\vsubset{(\u - \x)}{k}}_2^2\\
	\stackrel{\eqref{eq:sum_of_subproblems_reductionNewGH_CV:1}}{\le}&
	\frac{4\beta}{ 2\tau\sigma'} \sum_{k=1}^K \norm{\v_k - \BA\x}_2^2
	+ \frac{s^2 \sigma'\beta}{2\tau} \sum_{k=1}^K \norm{\BA\vsubset{(\u - \x)}{k}}_2^2
	\end{align*}
	Then
	\begin{align*}
	C \ge& sG_\mathcal{H}(\x, \{\v_k'\}_{k=1}^K)
	+ \frac{\mu_g}{2}(1-s)s\norm{\uv-\xv}^2
	-  \frac{s^2 (\sigma'+\beta\sigma')}{2\tau} \sum_{k=1}^{K}\norm{\BA (\uv-\xv)_{[k]}}^2 \\
	&- \frac{9\beta}{ 2\tau\sigma'}  \sum_{k=1}^K \norm{\v_k - \BA\x}_2^2
	\end{align*}
	Then let $\bar{\sigma}' := (1+\beta)\sigma'$ and $\eta:=\gamma(1-\Theta)(1-\alpha)$  we have
	\begin{align*}
	&\mathbb{E}[\vc{\ooa}{t} - \vc{\ooa}{t+1}
		- \alpha\frac{\gamma\sigma_1'}{2\tau}
		\sum_{k=1}^K \norm{\BA \Delta \vc{\vsubset{\x}{k}}{t} }_2^2
	] \\
	\ge&
	\eta \left(
	sG_\mathcal{H}(\vc{\xv}{t}; \{{\vc{\v_k'}{t}}\}_{k=1}^K)
	-  \frac{s^2 \bar{\sigma}'}{2\tau} \vc{R}{t}
	\right)
	- \frac{9\eta\beta}{ 2\tau\sigma'} \sum_{k=1}^K \norm{\vc{\v_k}{t} - \BA\vc{\x}{t} }_2^2
	\end{align*}
\end{proof}

The following lemma correlates the consensus violation with the size of updates
\begin{lemma}\label{lemma:relate_cv_with_updates}
    Let $c>0$ be any constant value. Define $\vc{\delta}{0} := \0$ and
    \begin{equation}\label{eq:relate_cv_with_updates:1}
        \vc{\delta}{t+1} := \beta \vc{\delta}{t} + c c_1 \sum_{k=1}^K\norm{\Delta\vc{\v}{t}_k}_2^2
    \end{equation}
    Then the consensus violation has an upper bound.
    \begin{equation}
        \sum_{k=1}^K\norm{\vc{\v}{t}_k -\vc{\v}{t}}_2^2
        \le e_1 \vc{\delta}{t}
    \end{equation}
    where $e_1:= (1-\beta) / c$.
\end{lemma}
\begin{proof}
	Let
	\begin{equation}
		a_t := \sum_{k=1}^K\norm{\Delta\vc{\v}{t}_k}_2^2,
		b_t := \sum_{k=1}^K\norm{\vc{\v}{t}_k -\vc{\v}{t}}_2^2
	\end{equation}
	We want to prove that
	\begin{align}\label{eq:relate_cv_with_updates:3}
		b_t \le e_1 \vc{\delta}{t}
	\end{align}
	First $t=0$, $b_0=\vc{\delta}{0}=0$. If the claim holds for time $t-1$, then
	$b_{t-1} \le e_1 \vc{\delta}{t-1}$. At time $t$, we have
    \begin{align}
        b_t 
        \stackrel{\eqref{eq:consensus_violation}}{\le}&
        \beta b_{t-1} + (1-\beta)c_1 a_{t-1} \\
        \le& \beta \frac{1-\beta}{c} \vc{\delta}{t-1}
		+ (1 - \beta) c_1 a_{t-1} \\
		\le& \frac{1-\beta}{c} (\beta \vc{\delta}{t-1} + c c_1 a_{t-1}) \\
		\stackrel{\eqref{eq:relate_cv_with_updates:1}}{\le}& e_1 \vc{\delta}{t}
    \end{align}
	Thus we proved the lemma. 
\end{proof}

\begin{lemma}\label{lemma:sum_of_subproblems_reductionNewGH_timeinvar_prime}
	Let $g_i$ be strongly convex with convexity parameter $\mu_g\ge 0$ with respect to the norm $\norm{\cdot}, \forall~i\in[n]$.
	Then for all iterations $t$ of outer loop, and $s\in[0, 1]$, it holds that
	\begin{equation}\label{eq:sum_of_subproblems_reductionNewGH_timeinvar_prime}
	\begin{split}
	&\mathbb{E}[\ooa(\vc{\x}{t}; \{\vc{\v_k}{t}\}_{k=1}^K)- \ooa(\vc{\x}{t+1}; \{\vc{\v_k}{t+1}\}_{k=1}^K)
	+ \frac{1+\beta}{2} \vc{\delta}{t} - \vc{\delta}{t+1}]\\
	\ge&
	\eta \left(
	s G_\mathcal{H}(\vc{\x}{t}; \{\textstyle \vc{\sum_{i=1}^K \mixingmat_{ki}\v_i}{t}\}_{k=1}^K)
	-  \frac{s^2 \bar{\sigma}'}{2\tau} \vc{R}{t}
	\right)
	\end{split}
	\end{equation}
	where $\alpha:=(1+\frac{(1-\beta)^2}{36(1+\Theta)\beta})^{-1} \in[0,1]$, $\eta:=\gamma(1-\Theta)(1-\alpha)$, $\bar{\sigma}':=(1 + \beta)\sigma'$
	and 
	\begin{equation}\label{eq:sum_of_subproblems_reductionNew_R1}
	\textstyle
	R^{(t)}:= - \frac{\tau\mu_g  (1-s)}{\bar{\sigma}'s}\norm{\vc{\uv}{t} - \vc{\xv}{t} }^2
	+ \sum_{k=1}^{K}\norm{\BA\vsubset{(\vc{\uv}{t} - \vc{\xv}{t})}{k}}^2
	\end{equation}
	for $\vc{\uv}{t}\in\R^n$ with 
	$\vc{\bar{\locgra}'}{t}:= \sum_{k=1}^K \nabla f(\sum_{i=1}^K \mixingmat_{ik} \vc{\v_i}{t})$
	\begin{equation}
	\vc{u_{i}}{t} \in \partial g_i^*(-\BA_i^\top \vc{\bar{\locgra}'}{t} ) \qquad k \text{ s.t. }i \in \Pk.
	\end{equation}
	where $\vc{\delta}{t}$ is defined in \cref{lemma:relate_cv_with_updates}.
\end{lemma}

\begin{proof}
	In this proof we use $\v_k':= \sum_i \mixingmat_{ki} \v_i$.
	From \cref{lemma:sum_of_subproblems_reductionNewGH_CV} we know that
	\begin{align*}
		&\mathbb{E}\bigg[\vc{\ooa}{t} - \vc{\ooa}{t+1}
			- \alpha\frac{\gamma\sigma_1'}{2\tau} \sum_{k=1}^K \norm{\BA \Delta \vc{\vsubset{\x}{k}}{t} }_2^2
            + \frac{9\eta\beta}{ 2\tau\sigma'} \sum_{k=1}^K \norm{\vc{\v_k}{t} - \BA\vc{\x}{t} }_2^2
		\bigg] \\
		\ge&
		\eta \left(
		sG_\mathcal{H}(\vc{\xv}{t}; \{{\vc{\v_k'}{t}}\}_{k=1}^K)
		-  \frac{s^2 \bar{\sigma}'}{2\tau} \vc{R}{t}
		\right)
	\end{align*}
    Use the following notations to simplify the calculation
    \begin{equation}
        a_t := \sum_{k=1}^K\norm{\Delta\vc{\v}{t}_k}_2^2,
        b_t := \sum_{k=1}^K\norm{\vc{\v}{t}_k -\vc{\v}{t}}_2^2,
        f_1 := \alpha\frac{\gamma \sigma_1'}{2\tau},
        f_2 := \frac{9\eta\beta}{2\tau \sigma'}
    \end{equation}
    From \cref{lemma:relate_cv_with_updates} we know that 
    \begin{equation}\label{eq:sum_of_subproblems_reductionNew_R1:1}
        f_2 b_t - f_1 a_t \le f_2 e_1 \vc{\delta}{t} - f_1 (\vc{\delta}{t+1} - \beta \vc{\delta}{t}) / (cc_1)
        = (f_2 e_1  + \frac{f_1\beta}{cc_1})\vc{\delta}{t} - \frac{f_1}{cc_1}\vc{\delta}{t+1}
    \end{equation}
    Fix constant $c$ such that $\frac{f_1}{cc_1}=1$ in \eqref{eq:sum_of_subproblems_reductionNew_R1:1}
    \begin{equation}\label{eq:sum_of_subproblems_reductionNew_R1:2}
        c = \frac{f_1}{c_1} = \alpha\frac{(1-\beta)^2 \sigma_1'}{2\tau \gamma K^2}
	\end{equation}
	Fix $(f_2 e_1  + \frac{f_1\beta}{cc_1}) = \frac{1+\beta}{2} < 1$
	in \eqref{eq:sum_of_subproblems_reductionNew_R1:1},
    to determine $\alpha\in[0, 1]$. First consider $f_2e_1$ 
    \begin{equation}
        f_2e_1 = \frac{9\gamma(1-\Theta)(1-\alpha)\beta}{2\tau \sigma'} \frac{1-\beta}{c}
		\stackrel{\eqref{eq:sum_of_subproblems_reductionNew_R1:2}}{=}
		\frac{1-\alpha}{\alpha}
		\frac{9(1-\Theta)\beta}{1-\beta} \frac{\gamma K}{\sigma_1'}
    \end{equation}
	Then we have 
    \begin{equation}
        f_2 e_1  + \frac{f_1\beta}{cc_1} = \frac{1-\alpha}{\alpha} \frac{9(1-\Theta)\beta}{1-\beta} \frac{\gamma K}{\sigma_1'} + \beta = \frac{1 + \beta}{2} < 1.
    \end{equation}
    Thus we can fix $\alpha\in [0, 1]$ to be
    \begin{equation}
		\alpha:=\bigg(1+\frac{(1-\beta)^2}{36(1+\Theta)\beta}\bigg)^{-1}
    \end{equation}
    So when have these information.
    \begin{equation}
        f_2 b_t - f_1 a_t \le \frac{1 + \beta}{2} \vc{\delta}{t} - \vc{\delta}{t+1}
    \end{equation}
    Finally, using all of the previous equations we know 
    \begin{equation}
        \mathbb{E}\bigg[\vc{\ooa}{t} - \vc{\ooa}{t+1}
            + \frac{1 + \beta}{2} \vc{\delta}{t} - \vc{\delta}{t+1}
        \bigg]
        \ge
        \eta \left(
        sG_\mathcal{H}(\vc{\xv}{t}; \{{\vc{\v_k'}{t}}\}_{k=1}^K)
        -  \frac{s^2 \bar{\sigma}'}{2\tau} \vc{R}{t}
        \right)
    \end{equation}
\end{proof}

\subsection{Main theorems}\label{ssec:theorems}
Here we present the proofs of \cref{theorem:linear_rate} and \cref{theorem:sublinear_rate}.

\begin{lemma}\label{lemma:R_bound}
	If $g_i^*$ are $L$-Lipschitz continuous for all $i\in[n]$, then
	\begin{equation}\label{eq:R_bound_1}
	\forall~t: \vc{R}{t} \le 4 L^2 \sum_{k=1}^{K} \sigma_k n_k=4L^2\sigma,
	\end{equation}
	where $	\sigma_k:= \max_{\vsubset{\xv}{k}\in\R^n}
	\norm{\vsubset{\BA}{k}\vsubset{\xv}{k}}^2 / \norm{\vsubset{\xv}{k}}^2$.
\end{lemma}
\begin{proof}
	For general convex functions, the strong convexity parameter is $\mu_g =0$, and hence the definition \eqref{eq:sum_of_subproblems_reductionNew_R1} of the complexity constant $\vc{R}{t}$ becomes
	\begin{align*}
	\vc{R}{t}
	=
	\sum_{k=1}^{K} \norm{
		\BA \vsubset{(\vc{\u}{t} - \vc{\xv}{t})}{k}
	} ^2
	=
	\sum_{k=1}^{K} \sigma_k
	\norm{
		\vsubset{(\vc{\u}{t} - \vc{\xv}{t})}{k}
	} ^2
	\le
	\sum_{k=1}^{K} \sigma_k
	|\Pk|4L^2
	= 4L^2 \sigma
	\end{align*}
	Here the last inequality follows from $L$-Lipschitz property of $g^*$.
\end{proof}

\TheoremLinearRate*

\begin{proof}
	If $g_i(\cdot)$ are $\mu_g$-strongly convex, one can use the definition of $\sigma_k$ and $\sigma_\text{max}$ to find
	\begin{align}\label{eq:TheoremLinearRate:1}
	\notag
	\vc{R}{t}
	\le&
	- \frac{\tau\mu_g(1-s)}{\bar{\sigma}'s}\norm{\vc{\uv}{t} - \vc{\xv}{t} }^2
	+ \sum_{k=1}^{K}\norm{\BA\vsubset{(\vc{\uv}{t} - \vc{\xv}{t})}{k}}^2\\
	\le&
	\left(
	- \frac{\tau\mu_g(1-s)}{\bar{\sigma}'s}
	+\sigma_\text{max}
	\right)
	\norm{\vc{\uv}{t} - \vc{\xv}{t} }^2 .
	\end{align}
    If we set
    \begin{equation}
        s_0 = \frac{\tau\mu_g}{\tau\mu_g+\sigma_\text{max} \bar{\sigma}'}
    \end{equation}
    then $\vc{R}{t} \le 0$. The duality gap has a lower bound 
	duality gap
    \begin{equation}
        G_\hh(\vc{\x}{t}, \{ \vc{\v_k'}{t} \}_{k=1}^K)
        \ge \ooa(\vc{\x}{t}, \{\vc{\v_k'}{t}\}_{k=1}^K) - \ooa^*
        \ge \ooa(\vc{\x}{t+1}, \{\vc{\v_k}{t+1}\}_{k=1}^K) - \ooa^*
    \end{equation}
    and use \cref{lemma:sum_of_subproblems_reductionNewGH_timeinvar_prime}, we have
	\begin{align}\label{eq:TheoremLinearRate:2}
	\E[\vc{\ooa}{t} - \vc{\ooa}{t+1} + \frac{1 + \beta}{2} \vc{\delta}{t} - \vc{\delta}{t+1}]
	\ge
	\eta s_0
    G_\hh
	\ge
	\eta s_0
	(\vc{\ooa}{t+1} - {\ooa}^\star)
	\end{align}
    Then
    \begin{equation}
        \E[\vc{\ooa}{t} - {\ooa}^\star + \frac{1+\beta}{2}\vc{\delta}{t}]
        \ge
        (1 + \eta s_0)\E[\vc{\ooa}{t+1} - {\ooa}^\star + \vc{\delta}{t+1}]
    \end{equation}

	Therefore if we denote $\vc{\varepsilon}{t}_\hh:= \vc{\ooa}{t} - \ooa^\star + \vc{\delta}{t}$ we have recursively that
	\begin{align*}
	\E [\vc{\varepsilon}{t}_\hh]
	\le 
    \left(1 - \frac{\eta s_0}{1 + \eta s_0} \right)^{t}
	\vc{\varepsilon}{0}_\hh
	\le
	\exp\left( - \frac{\eta s_0}{1 + \eta s_0} t \right)
	\vc{\varepsilon}{0}_\hh
	\end{align*}

	The right hand side will be smaller than some $\varepsilon_\hh$ if 
	\begin{align*}
	T \ge  \frac{1 + \eta s_0}{\eta s_0}
	\log \frac{\vc{\varepsilon}{0}_\hh}{\varepsilon_\hh}
	\end{align*}
	Moreover, to bound the duality gap $\vc{G_\hh}{t}$, we have
	\begin{align*}
	\eta s_0
	\vc{G_\hh}{t}
	\stackrel{\eqref{eq:TheoremLinearRate:2}}{\le}
	\E [ \vc{\ooa}{t} - \vc{\ooa}{t+1} + \frac{1+\beta}{2}\vc{\delta}{t} - \vc{\delta}{t+1}]
	\le 
	\E [ \vc{\ooa}{t} - {\ooa}^* + \vc{\delta}{t}]
	\end{align*}
	Hence if $\varepsilon_\hh \le \eta s_0 \varepsilon_{G_\hh}$ then
	$\vc{G_\hh}{t} \le \varepsilon_{G_\hh}$. Therefore after 
	\begin{align*}
	T \ge  \frac{1 + \eta s_0}{\eta s_0}
	\log \left(
	\frac{1}{\eta s_0}
	\frac{\vc{\varepsilon_\hh}{0}}{\varepsilon_{G_\hh}}
	\right)
	\end{align*}
	iterations we have obtained a duality gap less then $\varepsilon_{G_\hh}$.

\end{proof} 

\TheoremSublinearRate*
\begin{proof}
We write $\vc{\ooa}{t}$ instead of $\ooa(\vc{\x}{t}; \{\vc{\v_k}{t}\}_{k=1}^K)$ and $\ooa^\star$ instead of $\ooa({\x}^\star; \{{\v_k}^\star\}_{k=1}^K)$.
We begin by estimating the expected change of feasibility for $\ooa$. We can bound this above by using \cref{lemma:sum_of_subproblems_reductionNewGH_timeinvar_prime} and the fact that $\OB(\cdot)$ is always a lower bound for $-\OA(\cdot)$ and then applying \eqref{eq:R_bound_1} to find
\begin{align}\label{eq:general_convex:2}
(1 + \eta s)\E [ \vc{\ooa}{t+1} - \ooa^\star + \vc{\delta}{t+1}]
\le 
( \vc{\ooa}{t} -\ooa^\star + \vc{\delta}{t} ) 
+ \eta \frac{\bar{\sigma}' s^2}{2\tau} 4L^2\sigma
\end{align}
Use \eqref{eq:general_convex:2} recursively we have
\begin{equation}\label{eq:general_convex:3}
\E[\vc{\ooa}{t} -\ooa^\star + \vc{\delta}{t}]
\le 
(1 + \eta s)^{-t}
( \vc{\ooa}{0} -\ooa^\star + \vc{\delta}{0} ) 
+ s \frac{4L^2\bar{\sigma}'\sigma}{2\tau}
\end{equation}
We know that $\vc{\delta}{0}=0$. Choose $s=1$ and
\begin{equation}
    t=t_0:= \max \left\{0, \left\lceil \frac{1 + \eta}{\eta} \log 
    \frac{2\tau ( \vc{\ooa}{0} -\ooa^\star ) }
    {4L^2\sigma \bar{\sigma}'} \right\rceil\right\}
\end{equation}

leads to
\begin{align}\label{eq:general_convex:4}
\E [ \vc{\ooa}{t} - \ooa^\star  + \vc{\delta}{t}]
\le & \frac{4L^2\bar{\sigma}'\sigma}{\tau}
\end{align}
Next, we show inductively that
\begin{equation}\label{eq:general_convex:5}
\forall~t\ge t_0: \E [ \vc{\ooa}{t} - \ooa^\star + \vc{\delta}{t}]
\le
\frac{4L^2\bar{\sigma}'\sigma}{\tau (1+\frac{1}{2} \eta(t-t_0 ))}.
\end{equation}
Clearly, \eqref{eq:general_convex:4} implies that \eqref{eq:general_convex:5} holds for $t=t_0$. Assuming that it holds for any $t\ge t_0$, we show that it must also hold for $t+ 1$. Indeed, using
\begin{equation}\label{eq:general_convex:6}
s =\frac{1}{1+\frac{1}{2} \eta(t-t_0)} \in [0, 1],
\end{equation}
we obtain
\begin{align*}
\E [ \vc{\ooa}{t+1} - \ooa^\star + \vc{\delta}{t+1}]
\le \frac{4L^2\sigma\bar{\sigma}'}{\tau}
\underbrace{
\left(\frac{1 + \frac{1}{2} \eta(t-t_0) -\frac{1}{2}\gamma(1-\Theta) }
{(1 + \frac{1}{2} \eta(t-t_0))^2}\right)}_{D}
\end{align*}
by applying the bounds \eqref{eq:general_convex:2} and \eqref{eq:general_convex:5}, plugging in the definition of $s$ \eqref{eq:general_convex:6}, and simplifying. We upper bound the term $D$ using the fact that geometric mean is less
or equal to arithmetic mean:
\begin{align*}
D =& \frac{1}{1 + \frac{1}{2} \eta(t + 1 -t_0) }
\underbrace{ 
	\frac{(1 + \frac{1}{2} \eta(t + 1-t_0)) (1 + \frac{1}{2} \eta(t -1-t_0)) }
	{(1 + \frac{1}{2} \eta(t -t_0))^2}
}_{\le 1}\\
\le& \frac{1}{1 + \frac{1}{2} \eta(t +1-t_0)}.
\end{align*}

We can  apply the results of \cref{lemma:sum_of_subproblems_reductionNewGH_timeinvar_prime} to get
\begin{align*}
\eta s G_\hh(\vc{\x}{t}, \{ \vc{\v_k}{t} \}_{k=1}^K ) \le 
\vc{\ooa}{t} - \vc{\ooa}{t+1}+  \vc{\delta}{t} - \vc{\delta}{t+1}
\end{align*}
Define the following iterate
\begin{align*}
\bar{\x} := \frac{1}{T-T_0} \sum_{t=T_0+1}^{T-1} \vc{\x}{t},  
\bar{\v}_k := \frac{1}{T-T_0}\sum_{t=T_0+1}^{T-1} \vc{\v}{t}_k, 
\bar{\w}_k := \frac{1}{T-T_0} \sum_{t=T_0+1}^{T-1} \nabla f(\vc{\v_k}{t})
\end{align*}
use \cref{lemma:R_bound} to obtain
\begin{align*}
\E [G_\hh(\bar{\x}, \{ \bar{\v}_k \}_{k=1}^K,\{\bar{\w}_k \}_{k=1}^K )]
\le& \frac{1}{T-T_0}\sum_{t=T_0}^{T-1} \E [G_\hh \left(  \vc{\x}{t}, \{ \vc{\v_k}{t} \}_{k=1}^K  \right) ] \\
\le& \frac{1}{\eta s} \frac{1}{T-T_0} 
\E [ \vc{\ooa}{T_0} - \ooa^\star + \vc{\delta}{T_0} ] + \frac{4L^2\sigma\bar{\sigma}'s}{2\tau}
\end{align*}
If $T\ge \lceil \frac{1}{\eta} \rceil + T_0$ such that $T_0\ge t_0$ we have
\begin{align*}
\E [G_\hh(\bar{\x}, \{ \bar{\v}_k \}_{k=1}^K,\{\bar{\w}_k \}_{k=1}^K )]
\le& \frac{1}{\eta s} \frac{1}{T-T_0}  
\left(  \frac{4L^2\bar{\sigma}'\sigma}{\tau (1+\frac{1}{2} \eta (T_0-t_0 )} \right)
+ \frac{4L^2\sigma\bar{\sigma}'s}{2\tau}\\
=& \frac{4L^2\sigma\bar{\sigma}'}{\tau}\left( 
\frac{1}{\eta s} \frac{1}{T-T_0} \frac{1}{ (1+\frac{1}{2} \eta (T_0-t_0 ))}
+\frac{s}{2}\right).
\end{align*}
Choosing
\begin{equation}\label{eq:general_convex:7}
s=\frac{1}{(T-T_0)\eta } \in [0, 1 ]
\end{equation}
gives us
\begin{equation}\label{eq:general_convex:8}
\E [G_\hh(\bar{\x}, \{ \bar{\v}_k \}_{k=1}^K,\{\bar{\w}_k \}_{k=1}^K )] \le
\frac{4L^2\sigma\bar{\sigma}'}{\tau}\left(
\frac{1}{1+\frac{1}{2} \eta (T_0-t_0 )}
+\frac{1 }{2} \frac{1}{(T-T_0)\eta }\right).
\end{equation}
To have right hand side of \eqref{eq:general_convex:8} smaller then $\varepsilon_{G_\hh}$ it is sufficient to choose $T_0$ and $T$ such that 
\begin{align}
\label{eq:general_convex:9:1}
\frac{4L^2\sigma\bar{\sigma}'}{\tau}\left(
\frac{1}{1+\frac{1}{2} \eta (T_0-t_0 )} \right)
\le& \frac{1}{2} \varepsilon_{G_\hh} \\
\label{eq:general_convex:9:2}
\frac{4L^2\sigma\bar{\sigma}'}{\tau}\left(\frac{1 }{2}
\frac{1}{(T-T_0)\eta }\right)
\le& \frac{1}{2} \varepsilon_{G_\hh} 
\end{align}
Hence if $T_0  \ge t_0 + \frac{2}{\eta } \left(\frac{8L^2\sigma\bar{\sigma}'}{\tau\varepsilon_{G_\hh}}-1 \right)$ and
$ T\ge T_0 + \frac{4L^2\sigma\bar{\sigma}'}{\tau\varepsilon_{G_\hh}\eta }$ then \eqref{eq:general_convex:9:1} and \eqref{eq:general_convex:9:2} are satisfied.

\end{proof}

\localcertificate*
\begin{proof}
	If the $\w_k$ variable in the duality gap \eqref{eq:GH} is fixed to $\w_k = \locgra_k := \nabla f(\v_k)$, then 
	using the equality condition of the Fenchel-Young inequality on $f$, the duality gap can be written as follows
	\begin{equation}\label{eq:prop:1}
	G_\mathcal{H} (\x; \{\v_k\}_{k=1}^K) :=
	\sum_{k=1}^{K} \left( \langle \v_k, \locgra_k \rangle + 
	\sum_{i \in\Pk} g_i(\x_i) + g_i^*(-\BA_i^\top\bar{\locgra})
	\right)
	\end{equation}
	where $\bar{\locgra}=\frac{1}{K}\sum_{k=1}^K \locgra_k$ is the only term locally unavailable. 
	\begin{equation}\label{eq:prop:11}
	G_\mathcal{H} \le
	\sum_{k=1}^{K} \left( \langle \v_k, \locgra_k \rangle + 
	\sum_{i \in\Pk} g_i(\x_i) + g_i^*(-\BA_i^\top\locgra_k)
	\right) 
	+ \left|\sum_{k=1}^{K}\sum_{i\in\Pk} \left(g_i^*(-\BA_i^\top \bar{\locgra}) - g_i^*(-\BA_i^\top \locgra_k)\right)\right|
	\end{equation}
	If both terms in \eqref{eq:prop:11} are less than $\varepsilon/2$, then $G_\mathcal{H} \le \varepsilon$. Since the first term can be calculated locally, we only need for all $k=1,\ldots K$
	\begin{equation}\label{eq:prop:12}
	\langle \v_k, \locgra_k \rangle + 
	\sum_{i \in\Pk} g_i(\x_i) + g_i^*(-\BA_i^\top\locgra_k)  \le \frac{\varepsilon}{2K}.
	\end{equation}
	Consider the second term in  \eqref{eq:prop:11}. Compute the difference between  $g_i^*(-\BA_i^\top\bar{\locgra})$  and $g_i^*(-\BA_i^\top {\locgra}_k)$
	\begin{equation}\label{eq:prop:2}
	|g_i^*(-\BA_i^\top \bar{\locgra}) - g_i^*(-\BA_i^\top \locgra_k)|
	\le L |-\BA_i^\top (\bar{\locgra}- \locgra_k)| \le L \norm{\BA_i}_2 \norm{\bar{\locgra}- \locgra_k}_2
	\end{equation}
	where we use \cref{lemma:duality} and $L$-Lipschitz continuity.
	Then sum up coordinates $i \in \Pk$ on node $k$ 
	\begin{equation}\label{eq:prop:3}
	\left|\sum_{i\in\Pk} \left(g_i^*(-\BA_i^\top \bar{\locgra}) - g_i^*(-\BA_i^\top \locgra_k)\right)\right|
	\le L \norm{\bar{\locgra}- \locgra_k}_2 \sum_{i\in\Pk}  \norm{\BA_i}_2.
	\end{equation}
	Sum up \eqref{eq:prop:3} for all $k=1,\ldots, K$ and apply the Cauchy-Schwarz inequality
	\begin{align}\label{eq:prop:6}
	\left|\sum_{k=1}^{K}\sum_{i\in\Pk} \left(g_i^*(-\BA_i^\top \bar{\locgra}) - g_i^*(-\BA_i^\top \locgra_k)\right)\right|
	\le L \sqrt{\sum_{k=1}^{K}\norm{\bar{\locgra}- \locgra_k}_2^2}
	\sqrt{\sum_{k=1}^{K}(\sum_{i\in\Pk}  \norm{\BA_i}_2)^2} .
	\end{align}
	
	We will upper bound $\sum_{i\in\Pk}  \norm{\BA_i}_2$  and $\norm{\bar{\locgra}- \locgra_k}_2$ separately.
	First we have
	\begin{equation}\label{eq:prop:4}
	\sum_{i\in\Pk}  \norm{\BA_i}_2 \le \sqrt{n_k} \norm{\vsubset{\BA}{k}}_F \le n_k \norm{\vsubset{\BA}{k}}_{\infty,2}  
	\le n_k \sqrt{\sigma_k},
	\end{equation}
	where we write $\norm{.}_{\infty,2}$ for the largest Euclidean norm of a column of the argument matrix, and then used the definition of $\sigma_k$ as in  \eqref{eq:sigma}.
	Let us write $\BG:=[\locgra_{1}; \cdots; \locgra_{K}]$, $\BE:=\frac{1}{K}[\one;\cdots; \one]$, then
	apply Young's inequality with $\delta$
	\begin{align*}
	\textstyle \sum_{k=1}^K \norm{\locgra_{k} - \bar{\locgra}}_2^2 =& \norm{\BG - \BG\BE}_F^2 \\
	\le&  (1 + \frac{1}{\delta})\norm{\BG - \BG\mixingmat}_F^2 + (1+\delta)\norm{\BG\mixingmat - \BG\BE}_F^2 \\
	=&(1 + \frac{1}{\delta})\norm{\BG - \BG\mixingmat}_F^2 +  (1+\delta)\norm{\BG(\BI -\BE)(\mixingmat - \BE)}_F^2 \\
	\le&(1 + \frac{1}{\delta})\norm{\BG - \BG\mixingmat}_F^2 + (1+\delta)\beta^2 \norm{\BG(\BI -\BE)}_F^2 \\
	=&\textstyle
	(1 + \frac{1}{\delta}) \sum_{k=1}^K \norm{\locgra_k - \frac{1}{|\mathcal{N}_k|} \sum_{j\in \mathcal{N}_k}  \locgra_j}_2^2
	+ (1+\delta)\beta^2 \sum_{k=1}^K \norm{\locgra_{k} - \bar{\locgra}}_2^2 
	\end{align*}
	Take $\delta:=(1-\beta) / \beta$, then we have
	\begin{align}\label{eq:prop:5}
	\sum_{k=1}^K \norm{\locgra_{k} - \bar{\locgra}}_2^2 
	\le&  \frac{1}{1-\beta} \sum_{k=1}^K \norm{\locgra_k - \frac{1}{|\mathcal{N}_k|} \sum_{j\in \mathcal{N}_k}  \locgra_j}_2^2
	+ \beta \sum_{k=1}^K \norm{\locgra_{k} - \bar{\locgra}}_2^2 \notag\\
	\le& \frac{1}{(1-\beta)^2} \sum_{k=1}^K \norm{\locgra_k - \frac{1}{|\mathcal{N}_k|} \sum_{j\in \mathcal{N}_k}  \locgra_j}_2^2
	\end{align}
	We now use \eqref{eq:prop:4} and \eqref{eq:prop:5} and impose
	\begin{equation}\label{eq:prop:13}
	\frac{1}{(1-\beta)^2} \sum_{k=1}^K \norm{\locgra_k - \frac{1}{|\mathcal{N}_k|} \sum_{j \in \mathcal{N}_k}  \locgra_j}_2^2
	\sum_{k=1}^{K} n_k^2 \sigma_{k} \le \left(\frac{\varepsilon}{2L} \right)^2
	\end{equation}
	then \eqref{eq:prop:6} is less than $\varepsilon / 2$. Finally, \eqref{eq:prop:13} can be guaranteed by imposing the following restrictions for all $k=1,\ldots, K$
	\begin{equation}\label{eq:prop:10}
	\norm{\locgra_k - \frac{1}{|\mathcal{N}_k|} \sum_{j\in \mathcal{N}_k}   \locgra_j}_2^2
	\le (1-\beta)^2\frac{\varepsilon^2}{4KL^2} \left(\sum_{k=1}^{K} n_k^2 \sigma_{k} \right)^{-1}
	\end{equation}
\end{proof}


%
%
%

\section{Experiment details}\label{sec:experiment_settings}

In this section we provide greater details about the experimental setup and implementations. All the codes are written in PyTorch (0.4.0a0+cc9d3b2) with MPI backend \citep{paszke2017automatic}.
In each experiment, we run centralized CoCoA for a sufficiently long time until progress stalled; then use their minimal value as the approximate  optima.

\paragraph{\diging.} \diging is a distributed algorithm based on inexact gradient and a gradient tracking technique. 
\citep{nedic2017achieving} proves  linear convergence of \diging when the distributed optimization objective is strongly convex  over time-varying graphs with a fixed learning rate. 
  In this experiments, we only consider the time-invariant graph. The stepsize is chosen via a grid search. \citep{nedic2017achieving} mentioned that the  EXTRA algorithm \citep{shi2015extra} is  almost identical to that of the DIGing algorithm when the same stepsize is chosen for both algorithms, so
we only present with \diging here.


\paragraph{\cola.} We implement \cola framework with local solvers from Scikit-Learn \citep{scikit-learn}. Their ElasticNet solver  uses coordinate descent internally. We note that since the framework and theory allow any internal solver to be used, \cola could benefit even beyond the results shown by using existing fast solvers. 
We implement \cocoa as a special case of \cola.
The aggregation parameter $\gamma$ is fixed to $1$ for all experiments.

\paragraph{ADMM.} Alternating Direction Method of Multipliers (ADMM) \citep{boyd2011distributed} 
is a classical  approach in distributed optimization problems.
Applying ADMM to decentralized settings \citep{Shi:2014js} involves solving
\begin{align*}\textstyle
\min_{x_i, z_{i j}}\sum_{i=1}^{L} f_i (x_i) \qquad
\text{s.t. } x_i = z_{ij}, x_j = z_{ij}, \qquad \forall~(i, j) \in \mathcal{E}
\end{align*}
where $z_{ij}$ is an auxiliary variable imposing the consensus constraint on neighbors $i$ and $j$.
We therefore employ the coordinate descent algorithm to solve the local problem.
The number of coordinates chosen in each round is the same as that of \cola.
We choose the penalty parameter from the conclusion of \citep{Shi:2014js}.

\paragraph{Additional experiments.}
We provide additional experimental results here. First the \textit{consensus violation}
$ \sum_{k=1}^{K} \norm{\v_k - \v_c}_2^2$ curve for \cref{fig:convergence} is displayed in~\cref{fig:consensus_violation}. As we can see, the consensus violation starts with 0 and soon becomes very large, then gradually drops down. This is because we are minimizing the sum of $\vc{\ooa}{t}$ and $\vc{\delta}{t}$, see the proof of \cref{theorem:linear_rate}. 
Then another model under failing nodes is tested in \cref{fig:time_varying_graph_reset} where $\vsubset{\x}{k}$ are initialized to 0 when node $k$ leave the network.
Note that we assume the leaving node $k$ will inform its neighborhood and modify their own local estimates so that the rest nodes still satisfy $\frac{1}{\text{\#nodes} }\sum_{k} \v_k = \v_c$. This failure model, however, oscillates and does not converge fast.

\begin{figure}
	\begin{minipage}{0.48\textwidth}
		\centering
		\includegraphics[width=1.0\textwidth, keepaspectratio]{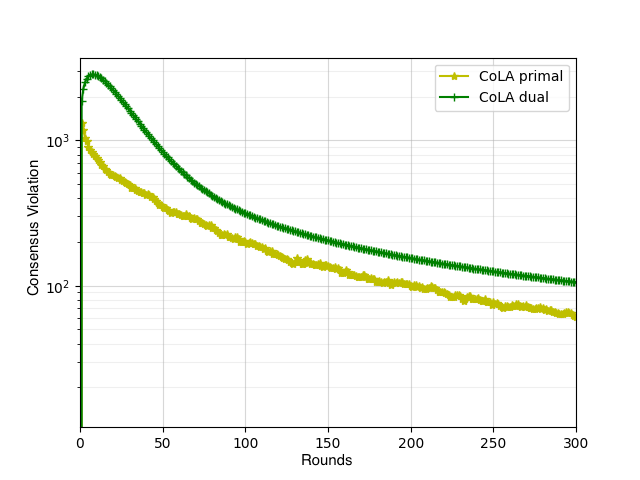}
		\caption{The consensus violation curve of \cola in \cref{fig:convergence}.\\~}
		\label{fig:consensus_violation}
	\end{minipage}
	~~~
	\begin{minipage}{0.48\textwidth}
		\centering
		\includegraphics[width=1.0\textwidth, keepaspectratio]{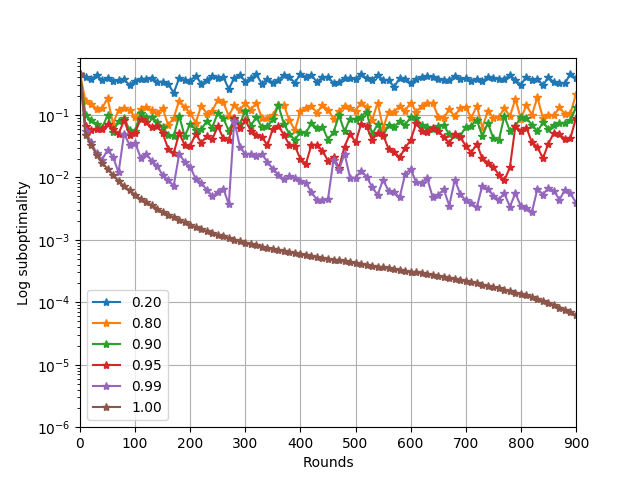}  
		\caption{Same settings as \cref{fig:time_varying_graph} except that $\vsubset{\x}{k}$ are reset after node $k$ leaving the network.\\~}
		\label{fig:time_varying_graph_reset}
	\end{minipage}
\end{figure}

\section{Details regarding extensions}
\label{app_extensions}

\subsection{Fault tolerance and time varying graphs} \label{ssub:ft_timevarying}
In this section we extend framework \cola to handle fault tolerance and time varying graphs.
Here we assume when a node leave the network, their local variables $\x$ are frozen.
We use same assumptions about the fault tolerance model in \citep{smith2017federated}. 
\begin{definition}[Per-Node-Per-Iteration-Approximation Parameter]\label{def:pnpiap}
	At each iteration $t$, we define the accuracy level of the solution calculated by node $k$ to its subproblem as
	\begin{equation}\label{eq:varying_theta}
	\theta_k^t:= \frac{\newsub(\vc{\Delta\xv}{t}_k; \vc{\vsubset{\xv}{k}}{t}, \vc{\v_k}{t})
		- \newsub({\Delta\xv}^\star_k; \vc{\vsubset{\xv}{k}}{t}, \vc{\v_k}{t})
	}{ \newsub(\0; \vc{\vsubset{\xv}{k}}{t}, \vc{\v_k}{t})
	- \newsub({\Delta\xv}^\star_k; \vc{\vsubset{\xv}{k}}{t}, \vc{\v_k}{t})}
	\end{equation}
	where ${\Delta\xv}^\star_k$ is the minimizer of the  subproblem $\newsub(\cdot; \vc{\vsubset{\xv}{k}}{t}, \vc{\v_k}{t})$.
	We allow this value to vary between [0, 1] with $\theta_k^t:=1$ meaning that no updates to subproblem 
	$\newsub$ are made by node $k$ at iteration $t$.
\end{definition}

The flexible choice of $\theta_k^t$ allows the consideration of stragglers and fault tolerance. We also need the following assumption on $\theta_k^t$.
\begin{assumption}[Fault Tolerance Model]\label{assumption:varying_theta}
	Let $\{ \vc{\x}{t} \}_{t=0}^T$ be the history of iterates until the beginning of iteration $T$.
	For all nodes $k$ and all iterations $t$, we assume
	$p_k^t := \P[\theta_k^t=1] \le p_\text{max} < 1$ and $\hat{\Theta}_k^T:=\E[\theta_k^T| \{ \vc{\x}{t} \}_{t=0}^T, \theta_k^T < 1] \le \Theta_\text{max}<1$.  
\end{assumption}
In addition we write $\bar{\Theta}:= p_\text{max} + (1 - p_\text{max}) \Theta_\text{max} < 1$. Another assumption on time varying model is necessary  in order to maintain the  same linear and sublinear convergence rate. It is from \cite[Assumption 1]{nedic2017achieving}: 
\begin{assumption}[Time Varying Model]\label{assumption:varying_beta} 
	Assume the mixing matrix $\mixingmat(t)$ is a function of time $t$.
	There exist a positive integer $B$ such that
	the spectral gap satisfies the following condition
	\begin{align*}\textstyle
	\sigma_{\text{max}}\left\{
	\prod_{i=t}^{t+B-1} \mixingmat(i) - \frac{1}{K} \one \one^\trans 
	\right\} \le \beta_\text{max} \qquad \forall~t \ge 0.
	\end{align*}
\end{assumption}
We change the \cref{alg_dcocoa} such that it performs gossip step for $B$ times between solving subproblems. In this way, the convergence rate  on time varying mixing matrix is similar to a static
graph with mixing matrix $\prod_{i=t}^{t+B-1} \mixingmat(i)$.
The sublinear/linear rate can be proved similarly.
\subsection{Data dependent aggregation parameter} \label{ssec:sigma}
\begin{definition}[Data-dependent aggregation parameter] In \cref{alg_dcocoa}, the aggregation parameter $\gamma$ controls the level of adding $\gamma$ versus averaging $\gamma := \frac{1}{K}$ of the partial solution from all machines. For the convergence discussed below to hold, the subproblem
	parameter $\sigma'$ must be chosen not smaller than
	\begin{equation}\label{eq:sigma_prime}
	\sigma' \ge \sigma'_\text{min} := \gamma \max_{\xv\in\R^n} \frac{\norm{\BA\xv}^2}{
		\sum_{k=1}^{K} \norm{\BA\vsubset{\xv}{k}}^2}
	\end{equation}
	The simple choice of $\sigma':=\gamma K$ is valid for \eqref{eq:sigma_prime}, closer to the actual bound given in $\sigma'_\text{min}$.
\end{definition}

\subsection{Hessian subproblem} \label{ssec:hessian}
If the Hessian matrix of $f$ is available, it can be used to define better local subproblems, as done in the classical distributed setting by  \citep{gargiani2017master,lee2017distributed,duenner2018trust,lee2018distributed}.  We use same idea in the decentralized setting, defining the improved  subproblem
\begin{equation}\label{eq:subproblem_modified_2}
	\begin{split}
		\newsub(\Delta\xv; \vsubset{\xv}{k}, \v_k)
			:=&\textstyle \frac{1}{K} f(\vv_k)
			+ \left\langle \sum_{l=1}^{K} \mixingmat_{kl}\nabla f(\vv_l), \BA\vsubset{\Delta\xv}{k} \right\rangle\\
			&\textstyle + \frac{1}{2}(\BA\vsubset{\Delta\xv}{k})^\top
				\left(\sum_{l=1}^{K}\mixingmat_{kl} \nabla^2 f(\vv_l)\right) \BA\vsubset{\Delta\xv}{k}
			+ \sum_{i \in \Pk} g_i(x_i + {\Delta x_i})
	\end{split}
\end{equation}
The sum of previous subproblems satisfies the following relations
\begin{align*}
	&\sum_{k=1}^{K}\newsub(\0; \vsubset{\vc{\xv}{t+1}}{k}, \vc{\v_k}{t+1})\\
	=&
	\frac{1}{K}\sum_{k=1}^{K} f\left(
		\sum_{l=1}^{K} \mixingmat_{kl} \vc{\vv_l}{t} + \gamma K \BA\vsubset{\Delta\xv}{k}
		\right)
	+
	\sum_{i\in\Pk} g_i(\vc{x_i}{t}+(\Delta\vsubset{\xv}{k})_i)\\
	\le&
	\frac{1}{K}
	\sum_{k=1}^{K}
	\sum_{l=1}^{K}
		\mixingmat_{kl}
		\left\{
			f(\vc{\vv_l}{t})
			+
			\langle \nabla f(\vc{\vv_l}{t}), \gamma K \BA\vsubset{\Delta\xv}{k}\rangle
			+
			\frac{1}{2} (\BA\vsubset{\Delta\xv}{k})^\top  \nabla^2 f(\vc{\vv_l}{t}) \BA\vsubset{\Delta\xv}{k}
		\right\}\\
	&+
	\sum_{i\in\Pk} g_i(\vc{x_i}{t}+(\Delta\vsubset{\xv}{k})_i)\\
	=&
	\sum_{k=1}^{K} \newsub(\Delta\xv; \vc{\vsubset{\xv}{k}}{t}, \vc{\v_k}{t})
	\le
	\sum_{k=1}^{K} \newsub(\0; \vc{\vsubset{\xv}{k}}{t}, \vc{\v_k}{t})
\end{align*}
This means that the sequence $\big\{ \sum_{k=1}^{K} \newsub(\0; \vc{\vsubset{\xv}{k}}{t}, \vc{\v_k}{t}) \big\}_{t=0}^\infty$ is monotonically non-increasing. Following the reasoning in this paper, we can have similar convergence guarantees for both strongly convex and general convex problems.
Formalizing all detailed implications here would be out of the scope of this paper, but the main point is that the second-order techniques developed for the \cocoa framework also have their analogon in the decentralized setting.

\end{document}

%% file: defs_nips18.tex
\usepackage{url}
\usepackage{hyperref}
\definecolor{cite_color}{RGB}{0, 0, 255}
\definecolor{link_color}{RGB}{153, 0,0}  
\definecolor{url_color}{RGB}{153, 102,  0}
\definecolor{emp_color}{RGB}{0,0,255}
\hypersetup{
   colorlinks,
   citecolor=cite_color,
   linkcolor=link_color,
   urlcolor=url_color
}

\usepackage{graphicx}
\usepackage{subcaption}
\graphicspath{{./figs/}}

\usepackage[export]{adjustbox}

\usepackage{epsfig}
\usepackage{amsmath,amssymb,amsthm}
\usepackage{mathrsfs}
\usepackage{multirow}
\usepackage{booktabs}
\usepackage{wrapfig}
\usepackage{enumerate}
\usepackage{bm}
\usepackage{tabularx}

\usepackage{mathtools}
\usepackage{thmtools}
\usepackage{thm-restate}

\usepackage[ ruled, linesnumbered]{algorithm2e}

\SetCommentSty{mycommfont}
\SetKwInOut{Input}{input}
\SetKwInOut{Output}{output}

\usepackage[capitalise, noabbrev]{cleveref} 
\crefname{section}{Section}{Sections}
\crefname{theorem}{Theorem}{Theorems}
\crefname{lemma}{Lemma}{Lemmas}
\crefname{equation}{Equation}{Equations}
\crefname{proposition}{Proposition}{Propositions}
\crefname{claim}{Claim}{Claims}
\crefname{assumption}{Assumption}{Assumptions}
\crefname{appendix}{Appendix}{Appendices}
\crefname{algorithm}{Algorithm}{Algorithms}
\crefname{figure}{Figure}{Figures}
\crefname{table}{Table}{Tables}
\crefname{remark}{Remark}{Remarks}
\crefname{definition}{Definition}{Definitions}
\crefname{equatinon}{Equation}{Equations}
\crefname{corollary}{Corollary}{Corollaries}

\allowdisplaybreaks

\newcommand{\appendixtitle}[1]{
    \begin{center}
        \LARGE \bf #1
    \end{center}
}

\def\E{{\mathbb E}}

\def \v{\mathbf{v}}
\def \w{\mathbf{w}}

\def \a{\mathbf{a}}
\def \b{\mathbf{b}}
\def \d{\mathbf{d}}
\def \x{\mathbf{x}}
\def \y{\mathbf{y}}
\def \s{\mathbf{s}}

\def \u{\mathbf{u}}

\def \z{\mathbf{z}}

\def \BA{\mathbf{A}}

\def \BI{\mathbf{I}}

\def \BE{\mathbf{E}}

\def \BG{\mathbf{G}}

\def \BV{\mathbf{V}}
\def \BW{\mathbf{W}}

\def \blambda{\bm{\lambda}}

\def \P{{\cal{P}}}

\def \R{{\mathbb{R}}}
\def \trans{\top}

\newcommand{\pare}[1]{{(#1)}}  

\def \E{{\mathbb{E}}} 

\newcommand{\argmin}{{\arg\min}}

\newcommand{\dom}[1]{{\texttt{dom}(#1)}}

\newcommand{\dtp}[2]{\langle #1, #2\rangle}

\newtheorem{assumption}{Assumption}

\newtheorem{lemma}{Lemma}

\newtheorem{definition}{Definition}
\newtheorem{remark}{Remark}
\newtheorem{claim}{Claim}

\newcommand{\zero}{\mathbf{0}} 
\newcommand{\one}{\mathbf{1}} 
\newcommand{\mixingmat}{\mathcal {W}}

\let \oldtextcircled \textcircled
\renewcommand{\textcircled}[1]{\oldtextcircled{\footnotesize #1}}

\newcommand{\cocoa}{\textsc{CoCoA}\xspace} 
\newcommand{\cola}{\textsc{CoLa}\xspace}
\newcommand{\diging}{{DIGing}\xspace}
\newcommand{\cocoap}{\textsc{CoCoA$\!^{\bf \textbf{\footnotesize+}}$}\xspace}

\newcommand{\wv}{ {\bf w}}
\newcommand{\xv}{ {\bf x}}
\newcommand{\vv}{ {\bf v}}
\newcommand{\uv}{ {\bf u}}

\newcommand{\OA}{F_{\hspace{-1pt}A}}
\newcommand{\OB}{F_{\hspace{-1pt}B}}

\newcommand{\ooa}{\mathcal{H}_{\hspace{-1pt}A}}
\newcommand{\oob}{\mathcal{H}_{\hspace{-1pt}B}}

\newcommand{\Pk}{\mathcal{P}_k}

\newcommand{\vsubset}[2]{#1_{[#2]}}
\newcommand{\eqdef}{:=}

\newcommand{\aggpar}{\gamma}
\newcommand{\vc}[2]{#1^{(#2)}}                   
\newcommand{\0}{ {\bf 0}}

\newcommand{\norm}[1]{\left\lVert{#1}\right\rVert}

\newcommand{\cconj}[1]{#1^*}

\renewcommand{\d}{d}
\newcommand{\n}{n}

\newcommand{\emphtitle}[1]{\textbf{#1}}

\newcommand{\newsub}{\mathscr{G}^{\sigma'}_k\hspace{-0.08em}}

\newcommand{\bigo}[1]{\mathcal O(#1)}
\newcommand{\locgra}{\mathbf{g}}
\newcommand{\hh}{\mathcal{H}} 

\newcommand{\colagap}{G_\mathcal{H}} 

%% file: neurips_2018.bbl
\begin{thebibliography}{45}
\providecommand{\natexlab}[1]{#1}
\providecommand{\url}[1]{\texttt{#1}}
\expandafter\ifx\csname urlstyle\endcsname\relax
  \providecommand{\doi}[1]{doi: #1}\else
  \providecommand{\doi}{doi: \begingroup \urlstyle{rm}\Url}\fi

\bibitem[Tsitsiklis et~al.(1986)Tsitsiklis, Bertsekas, and
  Athans]{Tsitsiklis:1986ee}
John~N Tsitsiklis, Dimitri~P Bertsekas, and Michael Athans.
\newblock {Distributed asynchronous deterministic and stochastic gradient
  optimization algorithms}.
\newblock \emph{IEEE Transactions on Automatic Control}, 31\penalty0
  (9):\penalty0 803--812, 1986.

\bibitem[Nedic and Ozdaglar(2009)]{nedic2009distributed}
Angelia Nedic and Asuman Ozdaglar.
\newblock Distributed subgradient methods for multi-agent optimization.
\newblock \emph{IEEE Transactions on Automatic Control}, 54\penalty0
  (1):\penalty0 48--61, 2009.

\bibitem[Duchi et~al.(2012)Duchi, Agarwal, and Wainwright]{duchi2012ddual}
J~C Duchi, A~Agarwal, and M~J Wainwright.
\newblock {Dual Averaging for Distributed Optimization: Convergence Analysis
  and Network Scaling}.
\newblock \emph{IEEE Transactions on Automatic Control}, 57\penalty0
  (3):\penalty0 592--606, March 2012.

\bibitem[Shi et~al.(2015)Shi, Ling, Wu, and Yin]{shi2015extra}
Wei Shi, Qing Ling, Gang Wu, and Wotao Yin.
\newblock Extra: An exact first-order algorithm for decentralized consensus
  optimization.
\newblock \emph{SIAM Journal on Optimization}, 25\penalty0 (2):\penalty0
  944--966, 2015.

\bibitem[Mokhtari and Ribeiro(2016)]{mokhtari2016dsa}
Aryan Mokhtari and Alejandro Ribeiro.
\newblock {DSA}: Decentralized double stochastic averaging gradient algorithm.
\newblock \emph{Journal of Machine Learning Research}, 17\penalty0
  (61):\penalty0 1--35, 2016.

\bibitem[Nedic et~al.(2017)Nedic, Olshevsky, and Shi]{nedic2017achieving}
Angelia Nedic, Alex Olshevsky, and Wei Shi.
\newblock Achieving geometric convergence for distributed optimization over
  time-varying graphs.
\newblock \emph{SIAM Journal on Optimization}, 27\penalty0 (4):\penalty0
  2597--2633, 2017.

\bibitem[Cevher et~al.(2014)Cevher, Becker, and Schmidt]{cevher2014review}
Volkan Cevher, Stephen Becker, and Mark Schmidt.
\newblock {Convex Optimization for Big Data: Scalable, randomized, and parallel
  algorithms for big data analytics}.
\newblock \emph{IEEE Signal Processing Magazine}, 31\penalty0 (5):\penalty0
  32--43, 2014.

\bibitem[Smith et~al.(2018)Smith, Forte, Ma, Tak{\'a}c, Jordan, and
  Jaggi]{smith2016cocoa}
Virginia Smith, Simone Forte, Chenxin Ma, Martin Tak{\'a}c, Michael~I Jordan,
  and Martin Jaggi.
\newblock {CoCoA: A General Framework for Communication-Efficient Distributed
  Optimization}.
\newblock \emph{Journal of Machine Learning Research}, 18\penalty0
  (230):\penalty0 1--49, 2018.

\bibitem[Zhang et~al.(2015)Zhang, Choromanska, and LeCun]{zhang2015elastic}
Sixin Zhang, Anna~E Choromanska, and Yann LeCun.
\newblock {Deep learning with Elastic Averaging SGD}.
\newblock In \emph{NIPS 2015 - Advances in Neural Information Processing
  Systems 28}, pages 685--693, 2015.

\bibitem[Wang et~al.(2017)Wang, Wang, and Srebro]{Wang:2017th}
Jialei Wang, Weiran Wang, and Nathan Srebro.
\newblock {Memory and Communication Efficient Distributed Stochastic
  Optimization with Minibatch Prox}.
\newblock In \emph{ICML 2017 - Proceedings of the 34th International Conference
  on Machine Learning}, pages 1882--1919, June 2017.

\bibitem[D{\"u}nner et~al.(2016)D{\"u}nner, Forte, Tak{\'a}c, and
  Jaggi]{Dunner:2016vga}
Celestine D{\"u}nner, Simone Forte, Martin Tak{\'a}c, and Martin Jaggi.
\newblock {Primal-Dual Rates and Certificates}.
\newblock In \emph{ICML 2016 - Proceedings of the 33th International Conference
  on Machine Learning}, pages 783--792, 2016.

\bibitem[Jaggi et~al.(2014)Jaggi, Smith, Tak{\'a}c, Terhorst, Krishnan,
  Hofmann, and Jordan]{jaggi2014communication}
Martin Jaggi, Virginia Smith, Martin Tak{\'a}c, Jonathan Terhorst, Sanjay
  Krishnan, Thomas Hofmann, and Michael~I Jordan.
\newblock Communication-efficient distributed dual coordinate ascent.
\newblock In \emph{Advances in Neural Information Processing Systems}, pages
  3068--3076, 2014.

\bibitem[Scaman et~al.(2017)Scaman, Bach, Bubeck, Lee, and
  Massouli{\'{e}}]{icmlScamanBBLM17}
Kevin Scaman, Francis~R. Bach, S{\'{e}}bastien Bubeck, Yin~Tat Lee, and Laurent
  Massouli{\'{e}}.
\newblock Optimal algorithms for smooth and strongly convex distributed
  optimization in networks.
\newblock In \emph{Proceedings of the 34th International Conference on Machine
  Learning, {ICML} 2017, Sydney, NSW, Australia, 6-11 August 2017}, pages
  3027--3036, 2017.

\bibitem[Scaman et~al.(2018)Scaman, Bach, Bubeck, Lee, and
  Massouli{\'e}]{scaman2018optimal}
Kevin Scaman, Francis Bach, S{\'e}bastien Bubeck, Yin~Tat Lee, and Laurent
  Massouli{\'e}.
\newblock Optimal algorithms for non-smooth distributed optimization in
  networks.
\newblock \emph{arXiv preprint arXiv:1806.00291}, 2018.

\bibitem[Jakoveti{c} et~al.(2012)Jakoveti{c}, Xavier, and
  Moura]{jakovetic2012convergence}
Dusan Jakoveti{c}, Joao Xavier, and Jos{e}~MF Moura.
\newblock Convergence rate analysis of distributed gradient methods for smooth
  optimization.
\newblock In \emph{Telecommunications Forum (TELFOR), 2012 20th}, pages
  867--870. IEEE, 2012.

\bibitem[Yuan et~al.(2016)Yuan, Ling, and Yin]{yuan2016convergence}
Kun Yuan, Qing Ling, and Wotao Yin.
\newblock On the convergence of decentralized gradient descent.
\newblock \emph{SIAM Journal on Optimization}, 26\penalty0 (3):\penalty0
  1835--1854, 2016.

\bibitem[Shi et~al.(2014)Shi, Ling, Yuan, Wu, and Yin]{Shi:2014js}
Wei Shi, Qing Ling, Kun Yuan, Gang Wu, and Wotao Yin.
\newblock {On the Linear Convergence of the ADMM in Decentralized Consensus
  Optimization}.
\newblock \emph{IEEE Transactions on Signal Processing}, 62\penalty0
  (7):\penalty0 1750--1761, 2014.

\bibitem[Wei and Ozdaglar(2013)]{Wei:2013wy}
Ermin Wei and Asuman Ozdaglar.
\newblock {On the O(1/k) Convergence of Asynchronous Distributed Alternating
  Direction Method of Multipliers}.
\newblock \emph{arXiv}, July 2013.

\bibitem[Bianchi et~al.(2016)Bianchi, Hachem, and
  Iutzeler]{bianchi2016coordinate}
Pascal Bianchi, Walid Hachem, and Franck Iutzeler.
\newblock A coordinate descent primal-dual algorithm and application to
  distributed asynchronous optimization.
\newblock \emph{IEEE Transactions on Automatic Control}, 61\penalty0
  (10):\penalty0 2947--2957, 2016.

\bibitem[Lian et~al.(2017)Lian, Zhang, Zhang, Hsieh, Zhang, and
  Liu]{lian2017can}
Xiangru Lian, Ce~Zhang, Huan Zhang, Cho-Jui Hsieh, Wei Zhang, and Ji~Liu.
\newblock Can decentralized algorithms outperform centralized algorithms? a
  case study for decentralized parallel stochastic gradient descent.
\newblock In \emph{Advances in Neural Information Processing Systems}, pages
  5336--5346, 2017.

\bibitem[Lian et~al.(2018)Lian, Zhang, Zhang, and Liu]{lian2017asynchronous}
Xiangru Lian, Wei Zhang, Ce~Zhang, and Ji~Liu.
\newblock Asynchronous decentralized parallel stochastic gradient descent.
\newblock In \emph{ICML 2018 - Proceedings of the 35th International Conference
  on Machine Learning}, 2018.

\bibitem[Tang et~al.(2018{\natexlab{a}})Tang, Lian, Yan, Zhang, and
  Liu]{tang2018d}
Hanlin Tang, Xiangru Lian, Ming Yan, Ce~Zhang, and Ji~Liu.
\newblock D$^2$: Decentralized training over decentralized data.
\newblock \emph{arXiv preprint arXiv:1803.07068}, 2018{\natexlab{a}}.

\bibitem[Tang et~al.(2018{\natexlab{b}})Tang, Gan, Zhang, Zhang, and
  Liu]{tang2018decentralized}
Hanlin Tang, Shaoduo Gan, Ce~Zhang, Tong Zhang, and Ji~Liu.
\newblock Communication compression for decentralized training.
\newblock In \emph{NIPS 2018 - Advances in Neural Information Processing
  Systems}, 2018{\natexlab{b}}.

\bibitem[Wu et~al.(2018)Wu, Yuan, Ling, Yin, and Sayed]{wu2018decentralized}
Tianyu Wu, Kun Yuan, Qing Ling, Wotao Yin, and Ali~H Sayed.
\newblock Decentralized consensus optimization with asynchrony and delays.
\newblock \emph{IEEE Transactions on Signal and Information Processing over
  Networks}, 4\penalty0 (2):\penalty0 293--307, 2018.

\bibitem[Sirb and Ye(2018)]{sirb2018decentralized}
Benjamin Sirb and Xiaojing Ye.
\newblock Decentralized consensus algorithm with delayed and stochastic
  gradients.
\newblock \emph{SIAM Journal on Optimization}, 28\penalty0 (2):\penalty0
  1232--1254, 2018.

\bibitem[Yang(2013)]{yang2013disdca}
Tianbao Yang.
\newblock {Trading Computation for Communication: Distributed Stochastic Dual
  Coordinate Ascent}.
\newblock In \emph{NIPS 2014 - Advances in Neural Information Processing
  Systems 27}, 2013.

\bibitem[Ma et~al.(2015)Ma, Smith, Jaggi, Jordan, Richt{\'a}rik, and
  Tak{\'a}c]{ma2015adding}
Chenxin Ma, Virginia Smith, Martin Jaggi, Michael~I Jordan, Peter
  Richt{\'a}rik, and Martin Tak{\'a}c.
\newblock {Adding vs. Averaging in Distributed Primal-Dual Optimization}.
\newblock In \emph{ICML 2015 - Proceedings of the 32th International Conference
  on Machine Learning}, pages 1973--1982, 2015.

\bibitem[D\"unner et~al.(2018)D\"unner, Lucchi, Gargiani, Bian, Hofmann, and
  Jaggi]{duenner2018trust}
Celestine D\"unner, Aurelien Lucchi, Matilde Gargiani, An~Bian, Thomas Hofmann,
  and Martin Jaggi.
\newblock {A Distributed Second-Order Algorithm You Can Trust}.
\newblock In \emph{ICML 2018 - Proceedings of the 35th International Conference
  on Machine Learning}, pages 1357--1365, July 2018.

\bibitem[Agarwal and Duchi(2011)]{agarwal2011distributed}
Alekh Agarwal and John~C Duchi.
\newblock Distributed delayed stochastic optimization.
\newblock In \emph{Advances in Neural Information Processing Systems}, pages
  873--881, 2011.

\bibitem[Zinkevich et~al.(2010)Zinkevich, Weimer, Li, and
  Smola]{zinkevich2010parallelized}
Martin Zinkevich, Markus Weimer, Lihong Li, and Alex~J Smola.
\newblock Parallelized stochastic gradient descent.
\newblock In \emph{Advances in Neural Information Processing Systems}, pages
  2595--2603, 2010.

\bibitem[Zhang and Lin(2015)]{zhang2015disco}
Yuchen Zhang and Xiao Lin.
\newblock Disco: Distributed optimization for self-concordant empirical loss.
\newblock In \emph{International conference on machine learning}, pages
  362--370, 2015.

\bibitem[Reddi et~al.(2016)Reddi, Kone{{c}}n{\`y}, Richt{\'a}rik,
  P{\'o}cz{\'o}s, and Smola]{reddi2016aide}
Sashank~J Reddi, Jakub Kone{{c}}n{\`y}, Peter Richt{\'a}rik, Barnab{\'a}s
  P{\'o}cz{\'o}s, and Alex Smola.
\newblock Aide: Fast and communication efficient distributed optimization.
\newblock \emph{arXiv preprint arXiv:1608.06879}, 2016.

\bibitem[Gargiani(2017)]{gargiani2017master}
Matilde Gargiani.
\newblock {Hessian-CoCoA: a general parallel and distributed framework for
  non-strongly convex regularizers}.
\newblock Master's thesis, ETH Zurich, June 2017.

\bibitem[Lee and Chang(2017)]{lee2017distributed}
Ching-pei Lee and Kai-Wei Chang.
\newblock Distributed block-diagonal approximation methods for regularized
  empirical risk minimization.
\newblock \emph{arXiv preprint arXiv:1709.03043}, 2017.

\bibitem[Lee et~al.(2018)Lee, Lim, and Wright]{lee2018distributed}
Ching-pei Lee, Cong~Han Lim, and Stephen~J Wright.
\newblock A distributed quasi-newton algorithm for empirical risk minimization
  with nonsmooth regularization.
\newblock In \emph{ACM International Conference on Knowledge Discovery and Data
  Mining}, 2018.

\bibitem[Smith et~al.(2017)Smith, Chiang, Sanjabi, and
  Talwalkar]{smith2017federated}
Virginia Smith, Chao-Kai Chiang, Maziar Sanjabi, and Ameet Talwalkar.
\newblock {Federated Multi-Task Learning}.
\newblock In \emph{NIPS 2017 - Advances in Neural Information Processing
  Systems 30}, 2017.

\bibitem[Kone{{c}}n{\`y} et~al.(2015)Kone{{c}}n{\`y}, McMahan, and
  Ramage]{konevcny2015federated}
Jakub Kone{{c}}n{\`y}, Brendan McMahan, and Daniel Ramage.
\newblock Federated optimization: Distributed optimization beyond the
  datacenter.
\newblock \emph{arXiv preprint arXiv:1511.03575}, 2015.

\bibitem[Kone{{c}}n{\`y} et~al.(2016)Kone{{c}}n{\`y}, McMahan, Yu, Richt{a}rik,
  Suresh, and Bacon]{konevcny2016federated}
Jakub Kone{{c}}n{\`y}, H~Brendan McMahan, Felix~X Yu, Peter Richt{a}rik,
  Ananda~Theertha Suresh, and Dave Bacon.
\newblock Federated learning: Strategies for improving communication
  efficiency.
\newblock \emph{arXiv preprint arXiv:1610.05492}, 2016.

\bibitem[McMahan et~al.(2017)McMahan, Moore, Ramage, Hampson, and
  y~Arcas]{mcmahan2017communication}
Brendan McMahan, Eider Moore, Daniel Ramage, Seth Hampson, and Blaise~Aguera
  y~Arcas.
\newblock Communication-efficient learning of deep networks from decentralized
  data.
\newblock In \emph{Artificial Intelligence and Statistics}, pages 1273--1282,
  2017.

\bibitem[Boyd et~al.(2011)Boyd, Parikh, Chu, Peleato, Eckstein,
  et~al.]{boyd2011distributed}
Stephen Boyd, Neal Parikh, Eric Chu, Borja Peleato, Jonathan Eckstein, et~al.
\newblock Distributed optimization and statistical learning via the alternating
  direction method of multipliers.
\newblock \emph{Foundations and Trends{\textregistered} in Machine learning},
  3\penalty0 (1):\penalty0 1--122, 2011.

\bibitem[Stich et~al.(2018)Stich, Cordonnier, and Jaggi]{stich2018sparsified}
Sebastian~U. Stich, Jean-Baptiste Cordonnier, and Martin Jaggi.
\newblock Sparsified sgd with memory.
\newblock In \emph{NIPS 2018 - Advances in Neural Information Processing
  Systems}, 2018.

\bibitem[Rockafellar(2015)]{rockafellar2015convex}
Ralph~Tyrell Rockafellar.
\newblock \emph{Convex analysis}.
\newblock Princeton university press, 2015.

\bibitem[Hastings(1970)]{hastings1970monte}
W~Keith Hastings.
\newblock Monte carlo sampling methods using markov chains and their
  applications.
\newblock \emph{Biometrika}, 57\penalty0 (1):\penalty0 97--109, 1970.

\bibitem[Paszke et~al.(2017)Paszke, Gross, Chintala, Chanan, Yang, DeVito, Lin,
  Desmaison, Antiga, and Lerer]{paszke2017automatic}
Adam Paszke, Sam Gross, Soumith Chintala, Gregory Chanan, Edward Yang, Zachary
  DeVito, Zeming Lin, Alban Desmaison, Luca Antiga, and Adam Lerer.
\newblock Automatic differentiation in pytorch.
\newblock In \emph{NIPS Workshop on Autodiff}, 2017.

\bibitem[Pedregosa et~al.(2011)Pedregosa, Varoquaux, Gramfort, Michel, Thirion,
  Grisel, Blondel, Prettenhofer, Weiss, Dubourg, Vanderplas, Passos,
  Cournapeau, Brucher, Perrot, and Duchesnay]{scikit-learn}
F.~Pedregosa, G.~Varoquaux, A.~Gramfort, V.~Michel, B.~Thirion, O.~Grisel,
  M.~Blondel, P.~Prettenhofer, R.~Weiss, V.~Dubourg, J.~Vanderplas, A.~Passos,
  D.~Cournapeau, M.~Brucher, M.~Perrot, and E.~Duchesnay.
\newblock Scikit-learn: Machine learning in {P}ython.
\newblock \emph{Journal of Machine Learning Research}, 12:\penalty0 2825--2830,
  2011.

\end{thebibliography}
